\documentclass[pra,twocolumn,superscriptaddress,floatfix]{revtex4}
\usepackage{graphicx,amsfonts,amssymb,amsmath, hyperref}

\usepackage[utf8]{inputenc}
\usepackage[T1]{fontenc}
\usepackage{lmodern}
\usepackage{csquotes}

\newif\ifhyper
\hypertrue
\ifhyper
\hypersetup{
   citecolor = {green},
   colorlinks = {true}, 
   urlcolor = {blue} 
}
\fi

\newcommand{\beq}{\begin{equation}}
\newcommand{\eeq}{\end{equation}}

\usepackage{mathtools}
\DeclarePairedDelimiter{\abs}{\lvert}{\rvert}
\DeclarePairedDelimiter{\ket}{\lvert}{\rangle}
\DeclarePairedDelimiter{\bra}{\langle}{\rvert}
\DeclarePairedDelimiterX{\braket}[2]{\langle}{\rangle}{#1\delimsize\vert #2}
\DeclarePairedDelimiterX{\ketbra}[2]{\lvert}{\rvert}{#1\delimsize\rangle\delimsize\langle #2}
\DeclarePairedDelimiterX{\expval}[3]{\langle}{\rangle}{#1\delimsize\vert #2\delimsize\vert #3}

\newcommand{\tr}{\mathop{}\mathopen{}\mathrm{tr}}

\usepackage{amsthm}
\newtheorem{lemma}{Lemma}
\newtheorem{theorem}{Theorem}
\newtheorem{observation}{Observation}

\theoremstyle{definition}
\newtheorem{definition}{Definition}

\let\ipr\braket

\begin{document}

\title{Geometric Entanglement in Topologically Ordered States}

\author{Rom\'an Or\'us}
\affiliation{Institute of Physics, Johannes Gutenberg University, 55099 Mainz, Germany}
\affiliation{Max-Planck-Institut f\"ur Quantenoptik, Hans-Kopfermann-Str. 1, 85748
Garching, Germany}
\affiliation{School of Mathematics and Physics, The University of Queensland,
QLD 4072, Australia}

\author{Tzu-Chieh Wei}
 \affiliation{C. N. Yang Institute for Theoretical Physics, State
University of New York at Stony Brook, NY 11794-3840, USA}

\author{Oliver Buerschaper}
\affiliation{Perimeter Institute for Theoretical Physics, 31 Caroline Street North, Waterloo, Ontario, Canada, N2L\,2Y5}

\author{Maarten Van den Nest}
\affiliation{Max-Planck-Institut f\"ur Quantenoptik,
Hans-Kopfermann-Str. 1, 85748 Garching, Germany}

\begin{abstract}
Here we investigate the connection between topological order and the
geometric entanglement, as measured by the logarithm of the overlap
between a given state and its closest product state of blocks. We do
this for a variety of topologically-ordered systems such as the
toric code, double semion, color code, and quantum double models. As
happens for the entanglement entropy, we find that for sufficiently
large block sizes the geometric entanglement is, up to possible
sub-leading corrections, the sum of two contributions: a bulk
contribution obeying a boundary law times the number of blocks, and
a contribution quantifying the underlying pattern of long-range
entanglement of the topologically-ordered state. This topological
contribution is also present in the case of single-spin blocks in
most cases, and constitutes an alternative characterisation of
topological order { for these quantum states} based on
a multipartite entanglement measure. In particular, we see that the
topological term for the 2D color code is twice as much as the one
for the toric code, in accordance with recent renormalization group
arguments [H. Bombin, G. Duclos-Cianci, D. Poulin, New J. Phys. 14
(2012) 073048]. Motivated by these results, we also derive a general
formalism to obtain upper- and lower-bounds to the geometric
entanglement of states with a non-Abelian group symmetry, and which
we explicitly use to analyse quantum double models. Furthermore, we
also provide an analysis of the robustness of the topological
contribution in terms of renormalization and perturbation theory
arguments, { as well as a numerical estimation for small systems}.
Some of the results in this paper rely on the ability to disentangle
single sites from the quantum state, which is always possible for
the systems that we consider. Additionally we relate our results to the
behaviour of the relative entropy of entanglement in
topologically-ordered systems, { and discuss a number of numerical
approaches based on tensor networks that could be employed to
extract this topological contribution for large  systems beyond
exactly-solvable models.}

\end{abstract}

\maketitle

\section{Introduction} Topological order (TO) \cite{to} is an
example of new physics beyond Landau's symmetry-breaking paradigm of
phase transitions. Systems exhibiting this new kind of order are
linked to concepts of the deepest physical interest, e.g.
quasiparticle anyonic statistics and topological quantum computation
\cite{topoq}. Importantly, TO finds a realisation in terms of
topological quantum field theories \cite{tqft}, which are the
low-energy limit of quantum lattice models such as the toric code
and quantum double models \cite{toric}, as well as string-net models
\cite{sn}.

A remarkable property about TO is that it influences the long-range
entanglement in the wave function of the system. For instance, as
proven for systems in two spatial dimensions (2D) \cite{Hamma,
entr}, the entanglement entropy $S$ of a region, such as $A$ in
Fig.~\ref{fig:boundaryL}.a, of boundary size $L \gg 1$ obeys the law
\beq
S = S_0 - S_{\gamma} + O(L^{-\nu}),
\eeq
where $S_0 \propto L$ is a
non-universal term (the so-called \enquote{boundary law}, or \enquote{area law}), $\nu$ some exponent, and
$S_{\gamma}$ is a universal long-distance contribution: the
topological entanglement entropy. $S_{\gamma}$ is non-zero for
systems with TO, e.g. $S_{\gamma} = 1$ for systems in the
topological phase of the toric code. In general one has that
\beq
S_\gamma = \log \left(\sqrt{\sum_a d_a^2}\right) ,
\eeq
where $\{ d_a \}$ are the so-called \emph{quantum dimensions} of the associated anyon model. More recently,
similar universal contributions have also been found for other
bipartite entanglement measures such as the mutual information and the R\'enyi entropy \cite{mut, reny}.

The main purpose of this paper is to investigate, in the context of
systems with TO, a global measure of entanglement which captures
multipartite correlations in the system: \emph{the geometric
entanglement (GE)} \cite{ge,brody}, which we shall denote by $E_G$.
This measure intuitively characterises how well an entangled state
can be approximated by a mean-field state, i.e., a product state, {
and its behaviour is in principle very different from that of
bipartite measures such as the entanglement entropy}.  Can such a
mean-field, global measure of entanglement be able to reveal any
property about topological order? As we shall see in this paper,
this is indeed the case.

Most of our study is focused on exactly-solvable models
corresponding to fixed-points of the renormalization group (RG). We
deal most prominently with the toric code model \cite{toric}, but
shall also discuss the double semion model \cite{sn}, topological
color codes \cite{Bombin}, as well as quantum double models
\cite{toric}. Each one of these models can be recast as some
particular instance of a string-net (Levin-Wen) model. Because of
this, they correspond (by definition) to RG fixed-points and,
therefore, they are representative of their respective topological
phases. Our main result is that, for all fixed-point models considered in this paper,
the GE obeys \beq \label{eqn:TGE} E_G = E_0 - E_{\gamma}, \eeq with
$E_{\gamma}$ being identical to $S_{\gamma}$ and $E_0$ some
{ bulk contribution}. The exact form of $E_0$ depends on
short-distance details such as the size of the blocks (i.e. the
number of spins) considered for the closest product state
optimisation. As we shall see, for \emph{blocks} of boundary size
$L$, $E_0$ will be proportional to $L$ and also to the number of
blocks $n_b$ (more specifically, $E_0 \propto n_b L$). We call this
behaviour, which depends on short-distance correlations, a
\enquote{boundary law} for the GE, since it corresponds to a
boundary term for each one of the blocks, i.e., $E_0/n_b \sim
L$. Moreover, $E_{\gamma}$ being identical to $S_{\gamma}$ is an
indication of the topological origin of the term $E_{\gamma}$.

Although we do not have a proof that Eq.~(\ref{eqn:TGE}) holds
generically beyond the considered fixed-point states, we shall argue
that the above law is also valid up to sub-leading corrections away
from the RG fixed points, e.g. for the toric code model under the
influence of non-relevant (and weak) perturbations such as magnetic
fields. In this case the geometric entanglement of blocks with
boundary size $L \gg 1$ seems to obey the law \beq E_G = E_0 -
E_{\gamma} + O(L^{-\nu'}), \label{toge} \eeq where again $E_0
\propto n_b L$ is a { bulk contribution}, $\nu '$ some
exponent, and $E_{\gamma}$ the topological contribution. From here
on, we shall refer to the term $E_{\gamma}$ as the \emph{topological
geometric entanglement}.

It is worth mentioning now a number of remarks about our results.
First, they provide the first evidence that topological order can
actually be read and assessed {in a number of situations} from
multipartite nature of entanglement, and, in particular, the
geometric measure of entanglement employed here. Thus, this may
provide an alternative way of characterising topological order in
quantum states. Second, our results are also the first explicit and
analytic examples of a boundary law behaviour for the GE in ground
states of 2D quantum many-body systems. Third, we will see that the
topological term for topological color codes is twice the one for
the toric code. This is consistent with the recent result that the
former model is equivalent to two copies of the latter
\cite{Bombin2}. Moreover, there may be potential
advantages in considering the GE to characterise topological order
instead of other quantities. For instance, one of them is that given
a (possibly approximate) description of the ground state by a Tensor
Network such as a Projected Entangled Pair State (PEPS) \cite{PEPS}
or a Multi-scale Entanglement Renormalization Ansatz (MERA)
\cite{MERA}, its computation is not too costly. This is because,  in
the context of Tensor Network methods, the calculation mainly
involves overlaps with product states and optimisations over these,
which can be computed quite efficiently. {In this respect, some of
the possible approaches using these methods will be discussed in the
Appendix}. Moreover, the GE can also be probed experimentally with
current technology using e.g. Nuclear Magnetic Resonance via single
spin measurement on all nuclear spins \cite{exper}, or ultra-cold
atoms in optical lattices via similar techniques to those explained
in Ref.~\cite{opticalLattice}.

The structure of the paper is as follows. First, we briefly review
the basics of the geometric entanglement in Sec.~\ref{sec:GE}. After
this, we deal in Sec.~\ref{sec:toric} with the toric code model.
Since this is the simplest model displaying non-trivial topological
order, we spend quite some time explaining many of its properties,
as well as derivations of the geometric entanglement of spins and
blocks both for the square and honeycomb lattices. In
Sec.~\ref{sec:beyond} we explain a number of important points that
need to be considered in order to generalise the results obtained
for the toric code to other models. Then, the double semion model is
analysed in Sec.~\ref{sec:semion}, and topological color codes are
addressed in Sec.~\ref{sec:color}. For color codes we derive results
using two independent methods: a bound approach (similar to the one
used for the other models), and a direct approach valid for
Calderbank-Shor-Steane (CSS) self-orthogonal codes. Depending on the
particular setting we will see that one approach may provide some
advantages over the other. After this, we develop in
Sec.~\ref{sec:general_bounds} a general formalism for the bound
approach, valid for states with non-Abelian symmetries, and apply
this formalism to quantum double models in
Sec.~\ref{sec:quantumdouble}. The situation beyond RG fixed points
is considered in Sec.~\ref{sec:beyondRG}, where the robustness under
perturbations of the topological contribution is assessed by RG and
perturbation theory arguments, {as well as by small-size numerical
calculations.} We present our conclusions and possible future
directions in Sec.~\ref{sec:conclude}. Finally, in the Appendix we
describe how our results are related to those of the relative
entropy of entanglement for topologically-ordered states, { and
discuss possible numerical algorithms based on tensor networks to
extract this topological component for non-exactly solvable models
of large size.}

\begin{figure}
 \includegraphics[width=8cm]{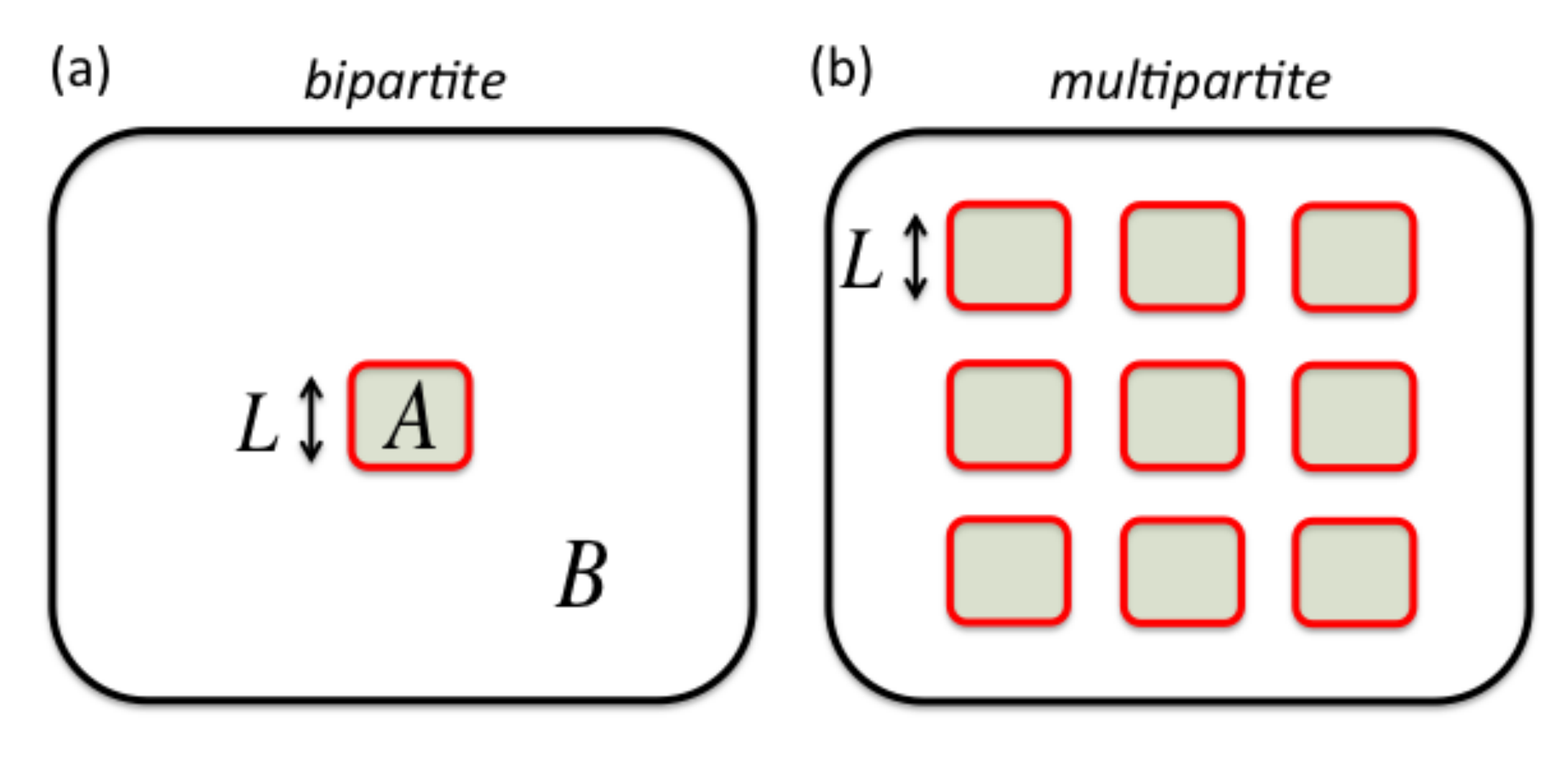}
  \caption{\label{fig:boundaryL}
  (color online) Two different types of partitions. (a) Bipartition with region $A$ and region $B$. Region $A$ has linear size $L$. This is usually the partition used in evaluations of the entanglement entropy and other bipartite entanglement measures. (b) Multipartition with many blocks, each one with linear size
  $L$. This is the partition used in the calculations of the geometric entanglement. If
  each block contains only one spin (akin to individual degree of freedom), then the product state is a
  completely separable state and the geometric entanglement is defined with respect to individual spins.
  However, each block can also contain multiple
  spins, and in this case the geometric entanglement is defined with respect to the closest product state of the individual
  blocks. In this latter scenario, we can inquire how the
  entanglement depends on the block boundary size $L$ and, in particular, whether there is something like a boundary law (i.e. linear scaling with $L$) and whether there is a topological correction.}
\end{figure}

\section{Geometric entanglement}\label{sec:GE} To begin with, let us recall some basic notions about the geometric entanglement. Consider an $m$-partite normalized
pure state $\ket{\Psi} \in \mathcal{H} = \bigotimes_{i=1}^{m}
\mathcal{H}^{[i]}$, where $\mathcal{H}^{[i]}$ is the Hilbert space
of party $i$. For instance, in a system of $n$ spins each party
could be a single spin, so that $m = n$. But each party could also be a set of
spins, either contiguous (a \emph{block} \cite{geometric2}) or not.

We wish now to determine how well state $\ket{\Psi}$ can be
approximated by an unentangled (normalized) state of the parties,
$\ket{\Phi}\equiv\mathop{\otimes}_{i=1}^{m}|\phi^{[i]}\rangle$. The
proximity of $\ket{\Psi}$ to $\ket{\Phi}$ is captured by their
overlap. The entanglement of $\ket{\Psi}$ is thus revealed by the
maximal overlap~\cite{ge},
\beq
\Lambda_{\max}({\Psi})\equiv\max_{\Phi}|\ipr{\Phi}{\Psi}|.
\eeq
The larger $\Lambda_{\max}$ is, the less entangled is $\ket{\Psi}$. Therefore, we
quantify the entanglement of $\ket{\Psi}$ via the quantity
\begin{equation}
E_G({\Psi})\equiv-\log_2\Lambda^2_{\max}(\Psi), \label{eq:Entrelate}
\end{equation}
where we have taken the base-2 logarithm, and which gives zero for unentangled states. $E_G(\Psi)$ is called \emph{geometric entanglement} (GE). This quantity has been studied in a variety of contexts, including critical systems and quantum phase transitions \cite{geometric2, geometric3}, quantification of entanglement as a resource for quantum computation \cite{resource}, local state discrimination \cite{discrim}, and has been recently measured in NMR experiments \cite{exper}.

To make a connection between the GE and other measures of entanglement, let us consider the case of just two sets of spins. In this case $E_G(\Psi)$ coincides with the so-called single-copy entanglement  between the two sets,
\begin{equation}
    \label{eq:single-copy}
    E_1(\Psi)
    =-\log_2
     \nu_1(\rho),
\end{equation}
with $\nu_1(\rho)$ the largest eigenvalue of the reduced density matrix $\rho$ of either set \cite{sc}. As is well known, this also coincides with the $\alpha$-R\'enyi entropy
\beq
S_{\alpha} = \frac{1}{(1-\alpha)} \log{ (\tr{(\rho^\alpha)})},
\eeq
for $\alpha \rightarrow \infty$.

Unlike other measures of entanglement, the GE offers a lot of flexibility to study multipartite quantum correlations in spin systems. For instance, one can choose each party to be a single spin, but one can also choose blocks of increasing boundary length $L$ \cite{geometric2} (see Fig.~\ref{fig:boundaryL}.b for an example). In fact, studying how the GE changes with $L$ provides valuable information about how close the system is to a product state under coarse-graining RG transformations.

\section{toric code Model} \label{sec:toric}
As described in the introduction, here we considered a variety of models corresponding to different topological phases. The simplest of these models is the \emph{toric code} \cite{toric}. This model is, in fact, the simplest example of a topologically non-trivial 2D system, and our results for this model provide also the grounds for the rest of the models that we shall consider (namely double semion, color codes, and quantum double models). Given its relevance for the rest of the paper, we provide now a short introduction to some of the key aspects of this model. As we shall see, this will be very useful for the forthcoming calculations.

Mathematically speaking, the toric code is the RG fixed point of the topological phase of a
$\mathbb{Z}_2$ gauge theory. The model is equivalent under local
transformations (or local \emph{moves}, or local \emph{disentanglers}) to the Levin-Wen string model on a honeycomb lattice
\cite{sn}, which is by definition a RG fixed point. A more precise derivation of this property will be provided later.

To define the toric code, we consider a lattice $\Sigma$ on a torus. In this paper we will deal with the toric code in the square and honeycomb lattices, yet the model can be defined on arbitrary lattices. Other Riemann surfaces of genus $\mathfrak{g}$ could also be considered easily without changing our conclusions. There are spin-1/2 (qubits) degrees of freedom attached to each link in lattice $\Sigma$. The model is described in terms of stars and plaquettes. A star \enquote{$s$} is a set of links sharing a common vertex. A plaquette \enquote{$p$} is an elementary face on the lattice $\Sigma$. For any star $s$ and plaquette $p$, we consider the star operators $A_s$ and plaquette operators $B_p$ defined as
\beq
A_s \equiv \prod_{j \in s} \sigma_x^{[j]} \ \ \ \ \ \ \ \ \ \ \ \ \  B_p \equiv \prod_{j \in p} \sigma_z^{[j]},
\eeq
where $\sigma_{\alpha}^{[j]}$ is the $\alpha$-th Pauli matrix at link $j$ of the lattice. Let us call respectively $n_s$, $n_p$ and $n$ the number of stars (vertices), plaquettes (faces) and links (sites) in lattice $\Sigma$. Importantly, star and plaquette operators satisfy the global constraint
\beq
\prod_s A_s = \prod_p B_p = \mathbb{I}.
\label{con}
\eeq
Therefore, there are $n_s-1$ independent star operators and $n_p-1$ independent plaquette operators.

With the definitions above, the Hamiltonian of the model reads
\beq
H_{{\rm TC}} = -\sum_s A_s -\sum_p B_p.
\eeq
This Hamiltonian is frustration free (i.e. all the terms in the above sum commute with each other), and can be diagonalized exactly as explained in Ref.~\cite{toric}. The ground level is $4$-fold degenerate on a torus ($4\mathfrak{g}$-fold for a Riemann surface of genus $\mathfrak{g}$). This degeneracy depends on the underlying topology of lattice $\Sigma$, and is already by itself a signature of topological order. Moreover, the ground level is a stabilised space of $\mathcal{G}_s$, the group of all the possible products of independent star operators, of size $|\mathcal{G}_s| = 2^{(n_s - 1)}$.

\subsection{Ground states}

In order to build  a basis for the ground level subspace, let us consider a closed curve $\gamma$ running on the links of the dual lattice $\Sigma^*$. We define its associated loop operator $W_x[\gamma] = \prod_{j \in \gamma} \sigma_x^{[j]}$, where $j \in \gamma$ are the links in $\Sigma$ crossed by the curve $\gamma$ connecting the centres of the plaquettes. It is not difficult to see that the group $\mathcal{G}_s$ can also be understood as the group generated by all the possible contractible loop operators \footnote{Notice that some sets of non-contractible loops are also elements of $\mathcal{G}_s$, such as two parallel non-contractible ones. This is because they can be built from products of contractible loops.}. Let also $\gamma_1$ and $\gamma_2$ be the two non-contractible loops on a torus, and define the associated string operators $w_{1,2} \equiv W_x[\gamma_{1,2}]$. We call $\ket{0}$ and $\ket{1}$ the eigenstates of $\sigma_z$ respectively  with $+1$ and $-1$ eigenvalue. With this assumptions, the ground level subspace then reads $\mathcal{L} = {\rm span} \{ \ket{i,j}, ~ i,j=0,1 \}$, where
\beq
\ket{i,j} = \frac{1}{\sqrt{|\mathcal{G}_s|}} \sum_{g \in \mathcal{G}_s} g ~ w_1^i ~ w_2^j ~ \ket{0}^{\otimes n}.
\label{gs}
\eeq
It is easy to check that the four vectors $\ket{i,j}$ are orthonormal and stabilised by $\mathcal{G}_s$, that is, $g \ket{i,j} = \ket{i,j} ~ \forall g \in \mathcal{G}_s$ and $\forall i,j$. These four states form a possible basis of the ground level subspace of the toric code model on a torus.

\subsection{Excited states}

Excited states of the toric code can be constructed by locally applying Pauli operators $\sigma_x, \sigma_y, \sigma_z$ on the ground states $\ket{i,j}$. Pauli operators $\sigma_z$ create pairs of deconfined charge-anticharge quasiparticle excitations,  $\sigma_x$ create pairs of deconfined flux-antiflux quasiparticles, and $\sigma_y$ creates a flux-antiflux and a charge-anticharge pairs. These operators can be applied to several sites of the lattice, thus creating a quasiparticle pattern that defines the excited state. The excited states of the model are then labeled by a set of quantum numbers, $\ket{\phi,c, i,j}$, where $\phi$ and $c$ are patterns denoting the position of flux-type and charge-type excitations respectively, and $i,j$ label the ground state $\ket{i,j}$ that was excited. In this paper we will focus on the entanglement properties of states $\ket{\phi,c, i,j}$.

\subsection{Disentangling the toric code}

\begin{figure}
\includegraphics[width=0.43\textwidth]{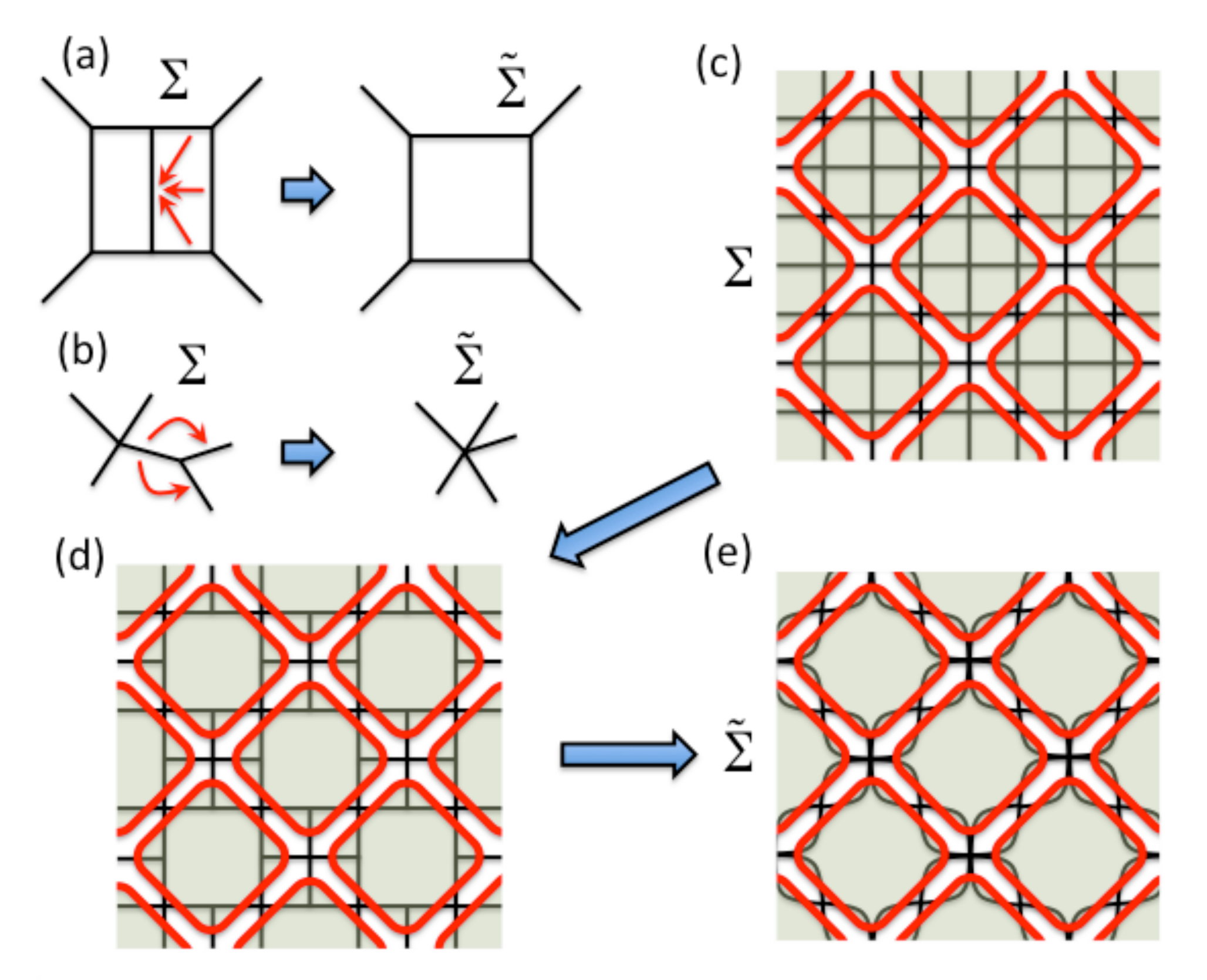}
\caption{(color online) (a,b) CNOT operations that disentangle qubits from the system, thus removing their links from the lattice. In the diagram, the arrows go from controlling to target qubits; (c) Example of a partition into blocks of a given boundary $L$, e.g. $L = 12$ here (but all our derivations work for arbitrary $L$). Qubits crossed by the boundary of a block are regarded as being inside the block; (d,e) Doing CNOT operations locally inside of each block we can remove all the stars inside of all the blocks in two steps: first, we apply CNOTs as in (a) to get (d) and, second, we apply CNOTs as in (b) to get (e). Notice that in (e) each star is made of either 4 or 8 qubits.  The procedure works similarly for other blockings.}
\label{fig:diag}
\end{figure}

A key property of the toric code model, which is fundamental for some of the derivations in this paper, is that its ground state $\ket{0,0} \equiv \ket{i=0,j=0}$ can be created by a quantum circuit that applies a sequence of Controlled-NOT (CNOT) unitary operations over an initial separable state of all the qubits \cite{cnots}. This means that it is actually possible to \emph{disentangle} qubits from the ground state of the system simply by reversing the action of these CNOTs. Moreover, these CNOT operations are \emph{local}, i.e. they do not span over delocalized sites in the lattice, but rather act over pairs of not-too-far-neighbouring qubits (e.g. nearest- and next-to-nearest-neighbours). This was the key observation that allowed to build an exact MERA representation of the ground states $\ket{i,j}$ of the model \cite{tcmera}. The two fundamental disentangling moves are represented in Fig~\ref{fig:diag}.a-b, and leave the overall quantum state as a product state of the disentangled qubits with the rest of the system. What is more, the rest of the system is left in the ground state of a toric code model on a deformed lattice $\widetilde{\Sigma}$, where $\widetilde{\Sigma}$ is obtained from $\Sigma$ by removing the links that correspond to the disentangled qubits. This property turns out to be of great importance for some of our derivations.

\subsection{GE of the toric code} \label{sec:GEtoric}
We are now in position to study the GE of the toric code model.
Given that this is a paradigmatic model of topological order, we
will perform a detailed analysis. First, we will provide a number of
lemmas and theorems that will bring us towards an expression for the
GE of the toric code in the square lattice, both for spins and
blocks,  in a particular basis. { Second, we will make a
similar study to the case of the honeycomb lattice. There we will
see that different choices of ground states may give rise to
different expressions for the GE of spins with the block size being
{\it unity}, yet the long-wavelength properties are the same and
equal to those for the square lattice after  {\it blocking of
spins\/} are considered. This is expected, as both lattices can be
mapped between each other by means of CNOT disentangling operations
(as shown in Fig.~\ref{figMap}), and therefore they should
correspond to the same RG fixed point. Here, we would like to
emphasize that blocking is essential as the topological contribution
is a long-wavelength property,  and, thus, that the behaviour of the
GE for sufficiently large block sizes is the same independently of
which of the two lattices we choose. (Without blocking, this
conclusion cannot be reached.)}

\begin{figure}
\includegraphics[width=0.4\textwidth]{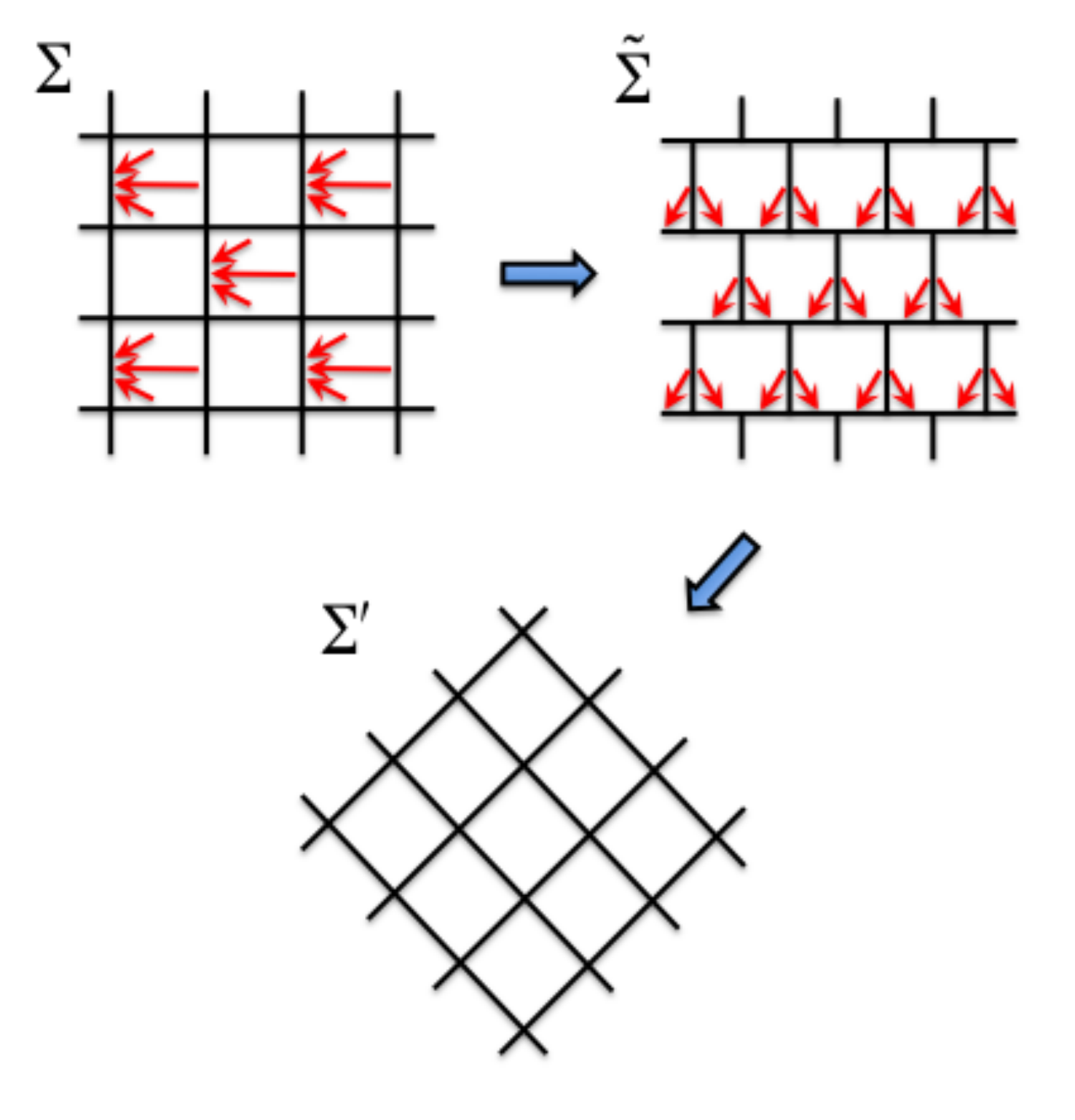}
\caption{(color online) Mapping of the toric code from a square
lattice $\Sigma$ to a honeycomb lattice $\widetilde{\Sigma}$
(Levin-Wen string model), and back to a square lattice $\Sigma'$.
Each red arrow represents a CNOT operation that disentangles qubits
from the system, as explained in Fig.~\ref{fig:diag}. The arrows go
from controlling to target qubits. } \label{figMap}
\end{figure}
\subsubsection{A couple of Lemmas}
Let us start by considering two Lemmas that will be quite useful in our derivations:

\begin{figure}
\includegraphics[width=0.43\textwidth]{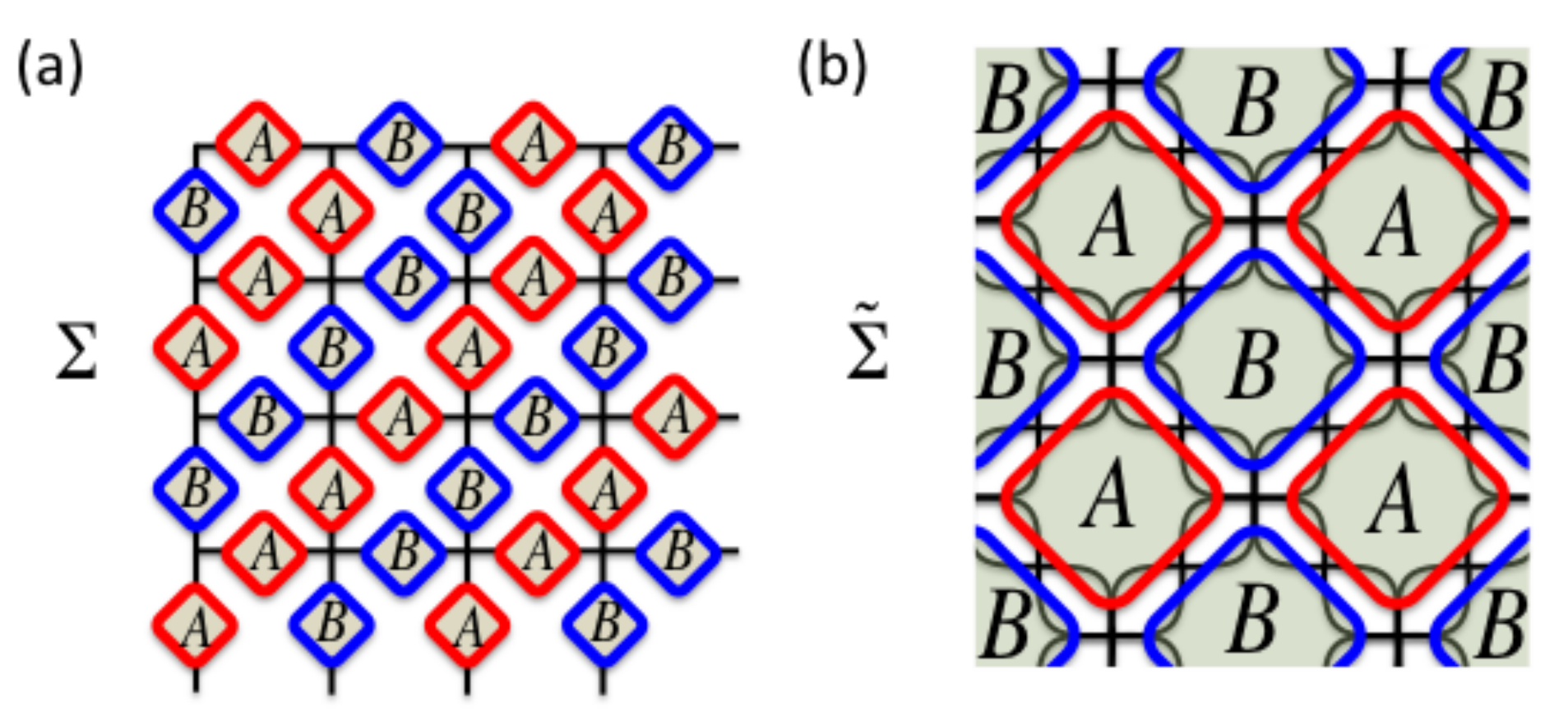}
\caption{(color online) Bipartite lattices of spins and blocks: (a) Sets of spins $A$ and $B$ for Theorem~\ref{thm:tc_spins} and a $4 \times 4$ torus; (b) Sets of blocks $A$ and $B$ for Theorem~\ref{thm:tc_blocks}.}
\label{part}
\end{figure}

\begin{lemma}
    \label{lem:tc_spins}
    All the eigenstates $\ket{\phi, c, i, j}$ of the
    toric code Hamiltonian have the same entanglement properties.
\end{lemma}

\begin{proof}
Let us choose a reference state from the ground level basis, e.g. $\ket{0,0}$. First of all, notice that the four basis states $\ket{i,j}$ in the ground level are related to $\ket{0,0}$ by local unitary operators,  since the string operators $w_1$ and $w_2$ are tensor products of $2 \times 2$ identity operators and Pauli-$x$ matrices. Since these operators act locally on each spin, they do not change the entanglement of the quantum state. Moreover, the excited states $\ket{\phi,c,i,j}$ are created locally by applying tensor products of $2 \times 2$ identity and Pauli-$x$ and $z$ operators to $\ket{i,j}$. The states resulting from these operations have then the entanglement properties of state $\ket{i,j}$, which in turn are the same as those of the reference state $\ket{0,0}$. Thus, all states $\ket{\phi, c, i, j}$ have the same entanglement properties.
\end{proof}

\begin{lemma}
    \label{lem:tc_blocks}
Consider an arbitrary set of blocks of qubits in
lattice $\Sigma$, where each block is regarded as an individual
party. Then, the entanglement properties of the ground state
$\ket{0,0}$ are the same as those of state $\ket{\widetilde{0,0}}$,
the ground state of the toric code model in the deformed lattice
$\widetilde{\Sigma}$ obtained after disentangling as many qubits as
possible using CNOTs locally inside of each block.
\end{lemma}

\begin{proof}
Given the ground state $\ket{0,0}$ in $\Sigma$, we start by doing CNOT disentangling operations \emph{locally inside of each block} in order to disentangle as many qubits as possible. As a result, state $\ket{0,0}$ transforms into state $\ket{0,0}_{{\rm disentangled}} \equiv \ket{e_1} \otimes \cdots \otimes \ket{e_p} \otimes \ket{\widetilde{0,0}}$, where  $\ket{e_k}$ is the quantum state for the $k$-th disentangled qubit, $k = 1,\ldots, p$ ($p$ is the number of disentangled qubits), and $\ket{\widetilde{0,0}}$ is the $i=0,j=0$ ground state of a toric code Hamiltonian in the deformed lattice $\widetilde{\Sigma}$. Since all CNOT operations are done locally inside of each block, states $\ket{0,0}$ and $\ket{0,0}_{{\rm disentangled}}$ have the same entanglement content if we regard the blocks as individual parties. What is more, the entanglement in the system is entirely due to the qubits in $\ket{\widetilde{0,0}}$, which proves the lemma.
\end{proof}

\subsubsection{Square lattice}
We are now in position to consider the entanglement properties of the four ground level states $\ket{i,j}$ as well as the excited states $\ket{\phi,c,i,j}$. Our main result here is that, for these states, the GE of blocks consists of a boundary term plus a topological contribution. Keeping the above two lemmas in mind, we now present the following theorem about the geometric entanglement of spins in the square-lattice toric code:

\begin{theorem}
    \label{thm:tc_spins}
For the toric code Hamiltonian in a square lattice $\Sigma$, the four ground states $\ket{i,j}$ for $i,j = 0,1$ and also the excited states $\ket{\phi, c, i, j}$ have all the same GE of spins and is given by
\beq
E_G = n_s - 1,
\label{getc}
\eeq
where $n_s$ is the number of stars in $\Sigma$.
\end{theorem}

\begin{proof}
Using Lemma~\ref{lem:tc_spins} we can restrict our attention to the GE for the state $\ket{0,0}$. In this setting, we call \emph{computational basis} the basis of the many-body Hilbert space constructed from the tensor products of the $\{ \ket{0}, \ket{1} \}$ local basis for every spin, which is an example of \emph{product basis} for the spins.

Our proof follows from upper and lower bounding the quantity $E_G(0,0)$ in Eq.(\ref{eq:Entrelate}) for the ground state $\ket{0,0}$. First, from the expression in Eq.(\ref{gs}) for the ground states we immediately have that the absolute value of the overlap with any state of the computational basis is $|\mathcal{G}_s|^{-1/2}$, and therefore $\Lambda_{\max}\ge |\mathcal{G}_s|^{-1/2}$. From here, we get $E_G(0,0) \le  \log_2|\mathcal{G}_s|$,
which gives an upper bound.

Next, to derive a lower bound
we use the fact that if we group the $n$ spins into two sets $A$ and $B$, then $\Lambda_{\max}\le \Lambda_{\max}^{[A:B]}$,
where $[A:B]$ means that a partition of the system with respect to the two sets $A$ and $B$ is considered. Thus, we have that  $E_G(0,0)\ge E_G(0,0)^{[A:B]}$.

The trick to find a useful lower bound is to find an appropriate choice of sets $A$ and $B$. In our case, we consider e.g. the bipartition shown in Fig.~\ref{part}.a for even $\times$ even lattices (other cases can be considered similarly). Then, we use a Lemma by Hamma, Ionicioiu and Zanardi in Ref.~\cite{Hamma},
that the reduced density matrix $\rho_A$ of $|0,0\rangle$ for subsystem $A$ satisfies $\rho_A^2= (|\mathcal{G}_s(A)| |\mathcal{G}_s(B)| / |\mathcal{G}_s| )\rho_A$,
where $\mathcal{G}_s(A/B)$ is the subgroup of $\mathcal{G}_s$ acting trivially on subsystem
$A/B$. Importantly, for $A$ and $B$ chosen as explained above, it happens that
$\mathcal{G}_s(A)$ and $\mathcal{G}_s(B)$ are \emph{trivial groups consisting of only the identity element}. With this in mind, we see that the reduced density matrix $\rho_A$ has eigenvalues either zero or
$|\mathcal{G}_s|^{-1}$, which is $|\mathcal{G}_s|$-fold degenerate. This means that
$({\Lambda}_{\max}^{[A:B]})^2=|\mathcal{G}_s|^{-1}$ and hence $E_G(0,0)\ge \log_2|\mathcal{G}_s|$.

Combining the two bounds, we get $E_G(0,0)= \log_2|\mathcal{G}_s|=n_s-1$, and from here Eq.(\ref{getc}) for the GE for spins follows immediately.
\end{proof}

Importantly, the techniques used in Theorem~\ref{thm:tc_spins} can also be used to deal with blocks of spins whenever the blocks form a bipartite lattice. In general, one may first disentangle as many qubits as possible inside the blocks. Then the remaining qubits are left in an entangled state where the GE is equal to the number of remaining independent star operators, which amounts to a boundary law term for each block (possibly with a sub-leading correction depending on how the blocks are chosen) plus a topological term. As an example of this, let us present the following theorem:

\begin{theorem}
    \label{thm:tc_blocks}
Given a partition of the square lattice $\Sigma$ into $n_b$ blocks of boundary size $L$ as indicated in Fig.~\ref{fig:diag}.c  (where $L$ is measured in number of qubits), then the four ground states $\ket{i,j}$ for $i,j = 0,1$ and also the excited states $\ket{\phi, c, i, j}$ of the toric code Hamiltonian have all the same GE of blocks and is given by
\beq
E_G = \frac{n_bL}{4} - 1.
\label{getcL}
\eeq
\end{theorem}

\begin{proof}
First, notice that Lemma~\ref{lem:tc_spins} and Lemma~\ref{lem:tc_blocks} imply that we can entirely focus on the ground state $\ket{\widetilde{0,0}}$ of a toric code model in the deformed lattice $\widetilde{\Sigma}$ from Fig.~\ref{fig:diag}.e. As shown in Fig.~\ref{fig:diag}.d-e, it is always possible to remove \emph{all} the stars inside of each block  in $\Sigma$ just by doing CNOTs locally inside of each block. As a result, the stars in the deformed lattice $\widetilde{\Sigma}$ correspond to those in $\Sigma$ that lay among the blocks. It is easy to see that, for $n_b$ blocks of boundary $L$, there are $\widetilde{n}_s = n_b L/4$ of such stars.

The rest of the proof follows as in Theorem~\ref{thm:tc_spins}, by upper and lower bounding $E_G(0,0)$ for the partition with respect to blocks. For the upper bound, we use the fact that ${\Lambda}_{\max}^{[{\rm blocks}]} \ge {\Lambda}_{\max}^{[{\rm spins}]} \ge |\widetilde{\mathcal{G}}_s|^{-1/2}$, where the first inequality again used the property that if larger blocks are considered then the maximum overlap is also larger, and $\widetilde{\mathcal{G}}_s$ is the corresponding group of contractible loop operators on $\widetilde{\Sigma}$. From here $E_G(0,0) \le  \log_2|\widetilde{\mathcal{G}}_s|$ follows. For the lower bound, we divide the system into two sets $A$ and $B$ of blocks as indicated in Fig.\ref{part}.b. Following a similar reasoning as in Theorem~\ref{thm:tc_spins}, it is possible to see again that no element $g \in \widetilde{\mathcal{G}}_s$ will act trivially on $A$ or $B$ except for the identity element. From this point, the rest of the proof is simply equivalent to the proof for Theorem~\ref{thm:tc_spins}.
\end{proof}

\subsubsection{Interpretation of the previous Theorems}
Let us now discuss Eqs.~(\ref{getc}) and (\ref{getcL}) in detail. First, let us remind that a variety of works have shown the existence of a \enquote{boundary law} for the entanglement entropy of many 2D systems \cite{boun}, including models with TO \cite{Hamma, entr}. It is then remarkable that, according to the first term of these equations, the entanglement per block obeys also a boundary law, i.e: it is proportional to the size $L$ of the boundary of the block. To the best of our knowledge, these results are the first example of a boundary law behaviour for a multipartite (rather than bipartite) measure of entanglement in 2D.

However, the second term in Eqs.~(\ref{getc}) and (\ref{getcL}) is far more intriguing and important. Its existence is caused by the global constraint from Eq.~(\ref{con}) on star operators which, in turn, allow for the topological degeneracy of the ground state. Thus, this term is \emph{of topological nature, and quantifies the pattern of long-range entanglement} present in topologically ordered states.

 In order to clarify further the meaning of the topological term, let us consider the bipartite case of a block of spins of boundary size $L$ and its environment (i.e. the rest of spins). In the bipartite case, the geometric entanglement $E_G$ coincides with the single-copy entanglement $E_1 = -\log \nu_1(\rho)$, with $\nu_1(\rho)$ the largest eigenvalue of the reduced density matrix $\rho$ of the block. Since this density matrix is proportional to a projector (see Ref.~\cite{Hamma} and also Ref.~\cite{reny}), we have that $E_1 = S$, with $S$ the entanglement entropy of the bipartition. Thus, for this case we have the chain of equalities
 \beq
E_G = E_1 = S.
\eeq
 We also know that for a system with TO, the entanglement entropy satisfies $S = S_0 - S_{\gamma} + O(L^{-\nu})$, where $S_0$ is some boundary law term and $S_{\gamma}$ is the topological entropy. As explained in Ref.\cite{Hamma}, for the toric code model the global constraints in Eq.(\ref{con}) imply that $S_{\gamma}=1$. Thus, with the convention from Eq.(\ref{toge}), we have that the \emph{topological geometric entanglement} is given, in this case, by
 \beq
 E_{\gamma} = S_{\gamma}.
 \eeq
The topological origin of $E_{\gamma}$ is thus clear. Notoriously, this topological contribution is maintained when promoting the entanglement measure from the bipartite to the multipartite scenario.

\subsubsection{Honeycomb lattice}

For the honeycomb lattice, most of the derivations are equivalent to
the square lattice, yet there is a small difference: while the
square lattice is self-dual, the honeycomb lattice is not. This
means, among other things, that the toric code in the honeycomb
lattice is not self-dual with respect to a change of star- and
plaquette-operators (which is actually the case for the square
lattice). In practice, this implies that different choices of ground
state basis may have different values of the GE of spins. However,
this is only a short-distance property, since  the GE of blocks has
the same topological contribution as for the square lattice.
In order to illustrate this more explicitly, we provide calculations
for two specific ground states.

\vspace{10pt}

\underline{\emph{The $|0,0\rangle$ ground state.-}}
We first consider the $\ket{0,0}$ ground state obtained as
\begin{equation}
|0,0\rangle= \frac{1}{\sqrt{|\mathcal{G}_s|}}\sum_{g\in \mathcal{G}_s} g \ket{0}^{\otimes n},
\end{equation}
which is nothing but the $i=0,j=0$ ground state in the notation of Eq.(\ref{gs}). This time, $\mathcal{G}_s$ is the group generated by the product of all possible star operators in the \emph{honeycomb} lattice. As before, $|\mathcal{G}_s| = 2^{n_s - 1}$, where $n_s$ is the number of stars (vertices) in the honeycomb lattice. For this state, we provide the following observation (not a Theorem):

\begin{observation}
For the toric code Hamiltonian in a honeycomb lattice $\Sigma$, the
GE of spins for the ground state $\ket{0,0}$ is upper- and
lower-bounded by \beq n_p-1 \le  {E_G} \le n_p+1,
\label{bo_honeycomb} \eeq where $n_p$ is the number of plaquettes
(faces) in $\Sigma$. Moreover, numerically it looks like the upper
bound is saturated, so that \beq {E_G} = n_p+1 . \label{obs} \eeq
\end{observation}
The above is an observation, and not a theorem,  since
Eq.(\ref{obs}) has only been checked numerically. Yet, let us prove
now the upper and lower bounds in Eq.(\ref{bo_honeycomb}). First,
the maximal overlap can be easily upper bounded again by
$\Lambda_{\max} \ge |\mathcal{G}_s|^{-1/2}=2^{(1-n_s)/2}$.
Nevertheless, unlike in the case of the square lattice, this time we
can actually have a better bound, i.e.,
\begin{eqnarray}
\Lambda_{\max}&\ge&\langle +|^{\otimes n}|0,0\rangle \nonumber \\
&=&\frac{1}{\sqrt{2^n}}\frac{1}{\sqrt{|\mathcal{G}_s|}}\sum_{x \in \{0,1\}^{
n}}\sum_{g\in \mathcal{G}_s} \langle x|g|0\rangle^{\otimes n} \nonumber \\
&=&\sqrt{\frac{|\mathcal{G}_s|}{2^n}}=\sqrt{\frac{2^{2n_p-1}}{2^{3n_p}}} = \frac{1}{\sqrt{2^{n_p + 1}}},
\end{eqnarray}
which is also smaller than $|\mathcal{G}_s|^{-1/2}=2^{(1-n_s)/2}$ using the fact that $n=3n_p$ and $n_s=2n_p$ for the honeycomb
lattice. This gives the desired upper bound on the GE, ${E_G}\le n_p+1$.

\begin{figure}
 \includegraphics[width=6cm]{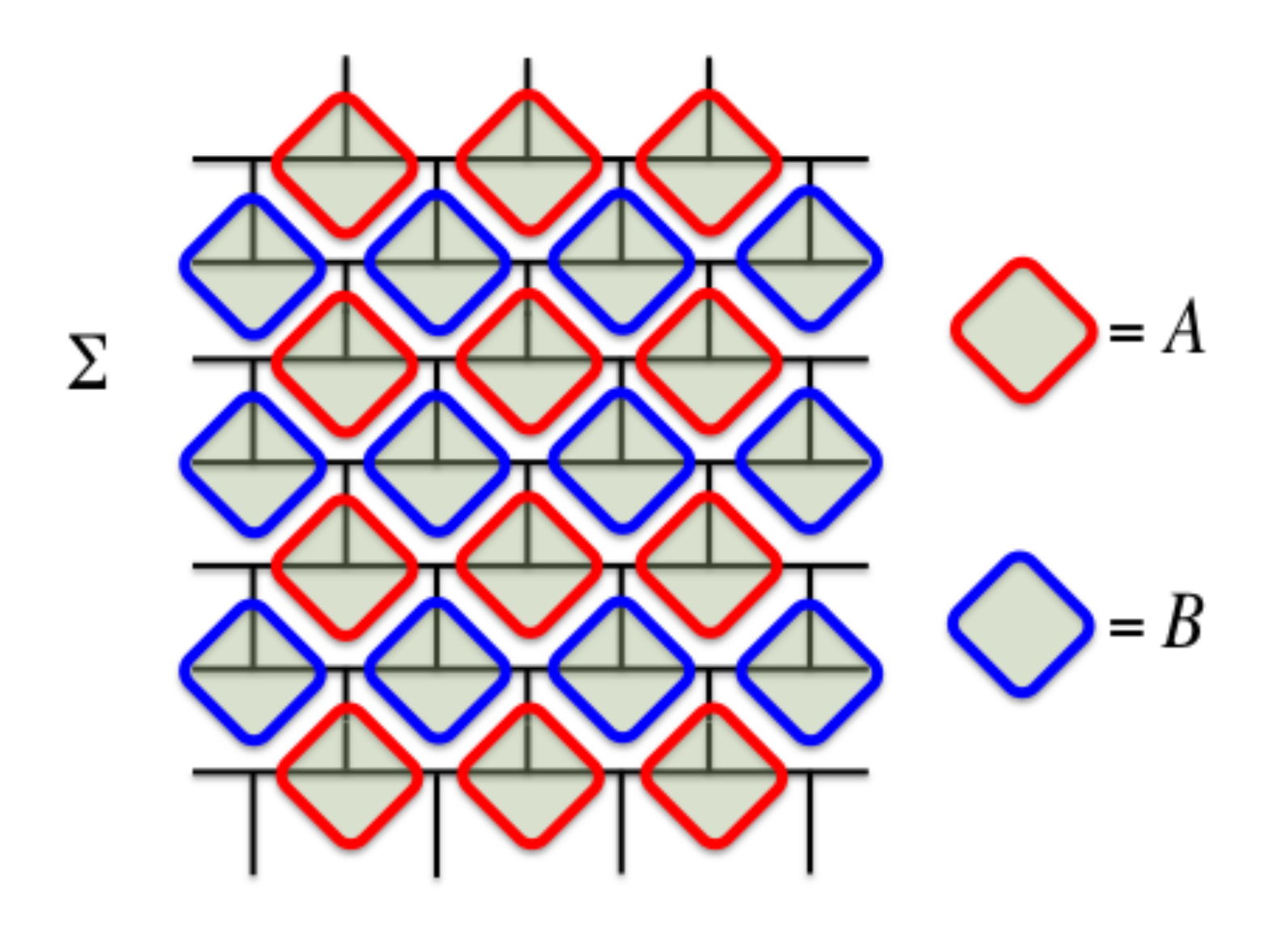}
  \caption{\label{fig:brickwallAB}
  Bipartition of edges on the honeycomb lattice (or equivalently the brickwall lattice) $\Sigma$. We have two sets: $A$ (red-grey) and $B$ (blue-white). Edges crossed by the block boundaries are regarded as inside the block.}
\end{figure}

To obtain the lower bound on GE (or equivalently upper bound on the
maximal overlap), our strategy is again similar to the one for the square lattice. This time we consider the bipartition illustrated in Fig.~\ref{fig:brickwallAB}. There is a star operator in each of the smallest units of the sets $A$ (red) and $B$ (blue). It is easy to convince oneself that the nontrivial subgroup $\mathcal{G}_s(A)$ of $\mathcal{G}_s$ is
generated by the $n_s/4$ star operators residing in all units in $A$,
and similarly for $\mathcal{G}_s(B)$. The size of these groups is therefore $|\mathcal{G}_s(A)| = |\mathcal{G}_s(B)| = 2^{n_s/4 }$. The reduced density matrix of e.g. set $A$ satisfies $\rho_A^2=(|\mathcal{G}_s(A)||\mathcal{G}_s(B)|/|\mathcal{G}_s|)\rho_A$, and hence the maximal overlap $\Lambda_{\max}^2$ is upper bounded by
\begin{equation}
\Lambda_{\max}^2 \le \frac{|\mathcal{G}_s(A)||\mathcal{G}_s(B)|}{|\mathcal{G}_s|} = \frac{1}{2^{n_p-1}},
\end{equation}
which leads to $n_p-1\le E_G$, as claimed. Combining this bound with the one computed previously, we arrive at the result that we mentioned before, namely, $n_p-1\le {E_G}\le n_p+1$.

In order to see whether any of these bounds is tight, we have
carried out a numerical computation for small system sizes and
confirmed that it is the upper bound which is saturated, i.e.,
${E_G}= n_p+1$. This can be seen in Table~\ref{tabd}. We would like
to emphasize that even though we do not obtain $-1$, this is only a
short-range property (where each block size is unity), as $E_G$ for
blocks will give rise to the correct contribution; see below.

\begin{table}
\begin{center}
\begin{tabular}{||c||c||c||c||c||}
\hline
$~~~N \times N~~~$ & $~~~n_s~~~$ & $~~~n~~~$ & $~~~n_p~~~$ &~~~ $E_G$ ~~~ \\
 \hline
 \hline
 $2 \times 2$ & 4 & 6 & 2 & 3 \\
 $2 \times 4$ & 8 & 12 & 4 & 5 \\
 $2 \times 6$ & 12 & 18 & 6 & 7 \\
 $2 \times 8$ & 16 & 24 & 8 & 9 \\
 $4 \times 4$ & 16 & 24 & 8 & 9 \\
 $4 \times 6$ & 24 & 36 & 12 & 13 \\
\hline
\end{tabular}
\end{center}
\caption{GE computed numerically for finite honeycomb lattices. We use a $N \times N$ brickwall structure to represent the honeycomb lattice with periodic boundary conditions, where $n_s = N^2, n = (3/2)N^2$, and $n_p = n - n_s$ (with even $N$). Our calculations are accurate with a $10^{-6}$ of relative error.}
\label{tabd}
\end{table}

\vspace{10pt}

\underline{\emph{The $|+,+\rangle$ ground state.-}}
We now consider the $\ket{+,+}$ ground state. By definition, this ground state is obtained by acting on the reference state $\ket{+}^{\otimes n}$ with the elements of the group $\mathcal{G}_p$ of plaquette operators, namely,
\begin{equation}
\label{eqn:++GS}
 |+,+\rangle\equiv \frac{1}{\sqrt{|\mathcal{G}_p|}}\sum_{g\in
\mathcal{G}_p} g \ket{+}^{\otimes n}.
\end{equation}
Notice that, unlike in the square lattice (where there is a duality between stars and plaquettes), in the honeycomb lattice the $\ket{+,+}$ ground state does not need to have, necessarily, the same properties as the $\ket{0,0}$ ground state from the previous section (at least not the same short-distance properties). For this state we can actually find its GE exactly, and provide the following Theorem:

\begin{theorem}
For the toric code Hamiltonian in a honeycomb lattice $\Sigma$, the
GE of spins for the ground state $\ket{+,+}$ is given by \beq {E_G}
= n_p-1 , \label{bo_honeycombdual} \eeq where $n_p$ is the number of
plaquettes (faces) in $\Sigma$.
\end{theorem}

\begin{proof}
We can obtain the GE analytically as follows: first, we use the duality that the sate $\ket{+,+}$ in the honeycomb lattice $\Sigma$ is equivalent to the $|0,0\rangle$ ground state of the toric code on the dual lattice $\Sigma^*$, which is the triangular lattice. After this, we simply follow the same approach as before by finding appropriate upper- and lower-bounds.  The trivial upper bound is given by $n_p - 1 \le E_G$, with $n_p$ the number of plaquettes in the original honeycomb lattice.
To find the lower bound, we consider the bipartition given in Fig.~\ref{fig:Triangular}, which is similar to the one that we used
square lattice (Fig.~\ref{thm:tc_blocks}). Using the same methods as before, the expression for the GE follows.
\end{proof}

\begin{figure}
 \includegraphics[width=9cm]{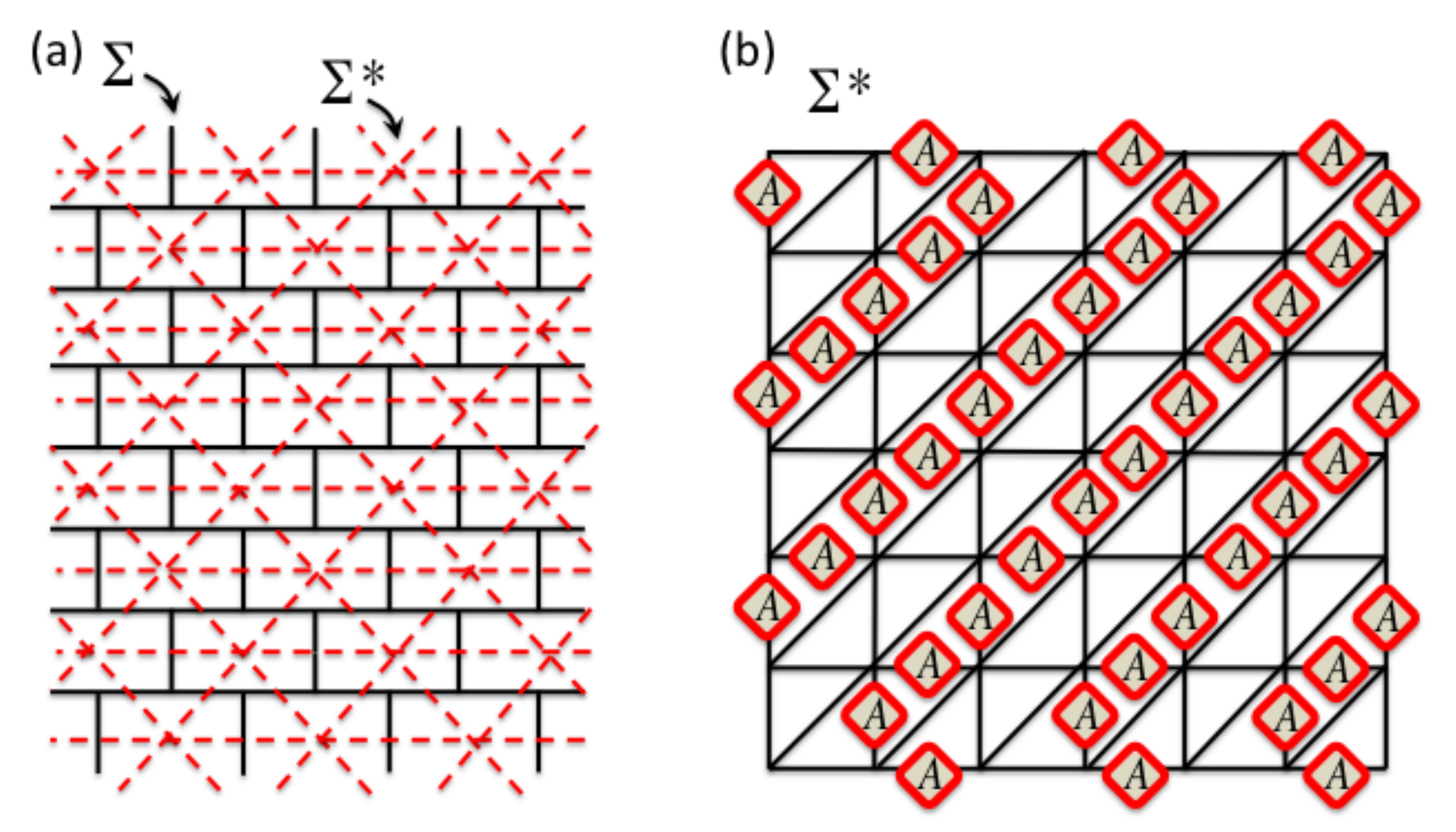}
  \caption{\label{fig:Triangular}
 (color online) (a) Honeycomb lattice (or brickwall lattice) $\Sigma$ in black lines, and its dual triangular lattice $\Sigma^*$ in dotted red lines. (b) Bipartition of edges in $A$ and $B$ sets on the triangular lattice (deformed to fit in to a square lattice). For simplicity of the figure, only the spins in the $A$ set are shown (in red). The rightmost  edges are identified with the leftmost edges, and the bottom edges are identified with the top edges, since periodic boundary conditions are used.}
\end{figure}

{

\underline{\emph{Blocking the honeycomb lattice.-}}  The fact that
the $|0,0\rangle$ state on different lattices  gives different
constant contributions to the GE of spins turns out to be a
short-distance property that disappears when blocks of spins are
considered. Blocking corresponds to coarse-graining and gives
long-wavelength properties.

This is easy to notice, if one realises that in just one RG step we
can map the honeycomb to the square lattice. Let us explain this in
more detail. It is well known that the toric code in the honeycomb
and square lattices can be mapped to one another by simple RG
(disentangling) operations acting locally on the lattice. These
transformations (from the square to the honeycomb lattice, and back
to the square) are represented in the diagram of Fig.~\ref{figMap}.
In fact, these moves can also be understood in terms of Entanglement
Renormalization steps \cite{tcmera}, where some of the initial
qubits become disentangled from the rest after every step.

Importantly, after mapping the honeycomb to the square lattice, the
GE follows  exactly the laws explained in the previous section, and
hence it is evident that $E_{\gamma} = 1$ also holds in this case
after a suitable blocking of the spins (so that the RG CNOTs fall
within the blocks).

For instance, one could choose the blocking shown in
Fig.~\ref{blockHon} and obtain $E_{\gamma} = 1$, but similar results
hold for other choices of blocks as well. One could also consider
larger blocking such as the one shown in
Fig.~\ref{fig:honeycomb_blocking}.  It is interesting to know that
one-step of blocking is sufficient to remove the short-range
behaviour and result in a correct topological contribution. Needless
to say, in the process of blocking the corresponding ``boundary
laws'' for the bulk contribution of the GE of blocks also follow. }

\begin{figure}
\includegraphics[width=0.45\textwidth]{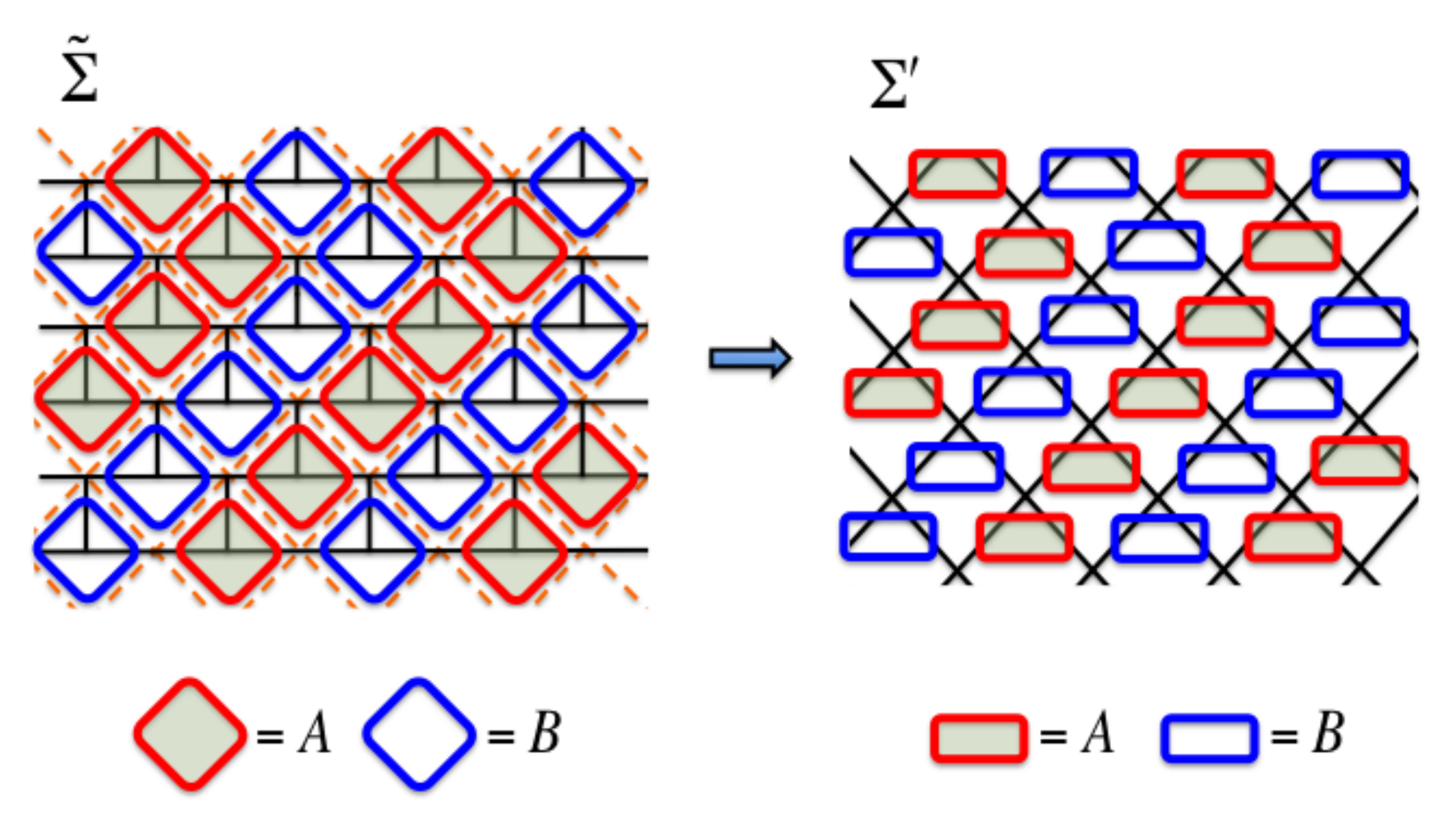}
\caption{(color online)  Blocking and bipartition of a honeycomb
lattice on a torus.  The blocking is such that the CNOT operations
from Fig.~\ref{figMap}, which map the honeycomb lattice (black) to
the square lattice (orange), fall within the blocks, so they do not
affect the calculation of the GE. The GE of 3-site blocks for the
honeycomb lattice (left) then maps to the GE of 2-site blocks for
the square lattice (right). Larger block sizes could also be
considered.}  \label{blockHon}
\end{figure}

\subsubsection{Brief discussion of the honeycomb lattice results}

\begin{figure}
   \includegraphics[height=0.3\textwidth]{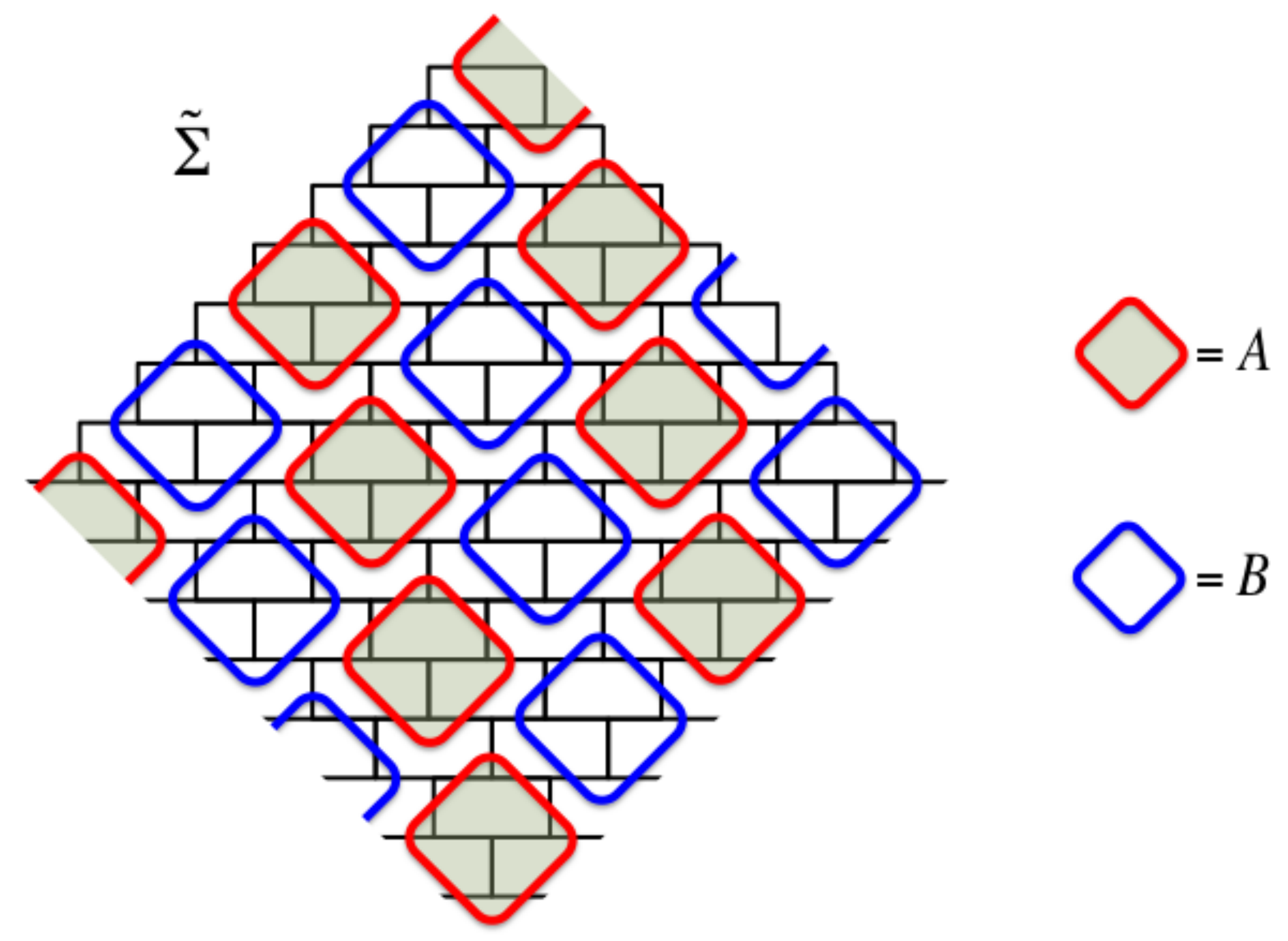}
  \caption{\label{fig:honeycomb_blocking} (color online) Blocking and bipartition of the honeycomb lattice on a torus, for larger blocks. Each block is defined by edges contained in red or blue square, where the coloring indicates the bipartition.  Larger blocks can be considered similarly. }
\end{figure}

Let us comment briefly on the above results. First, we have seen that the GE { of spins} is, in this case, different for the $\ket{0,0}$ and $\ket{+,+}$ states. One could also infer from here that the topological contribution to the GE is also different. However, we wish to stress that the topological GE is a long-range contribution, and can only be extracted reliably after appropriate blocking of the spins. { This is exactly what has been observed when considering the blocking of spins in the honeycomb lattice.}
Second, it is clear that the same conclusions that applied to excited states and other ground states on the square lattice apply also to the honeycomb lattice, i.e. once chosen a ground state basis, all the states in the basis as well as the corresponding (quasiparticle) excitations have the same amount of entanglement as the \enquote{reference} ground state (which in this case is either $\ket{0,0}$ or $\ket{+,+}$), as long as they are related by local operations.

\subsubsection{A different topology: the sphere}
\begin{figure}
   \includegraphics[height=0.35\textwidth]{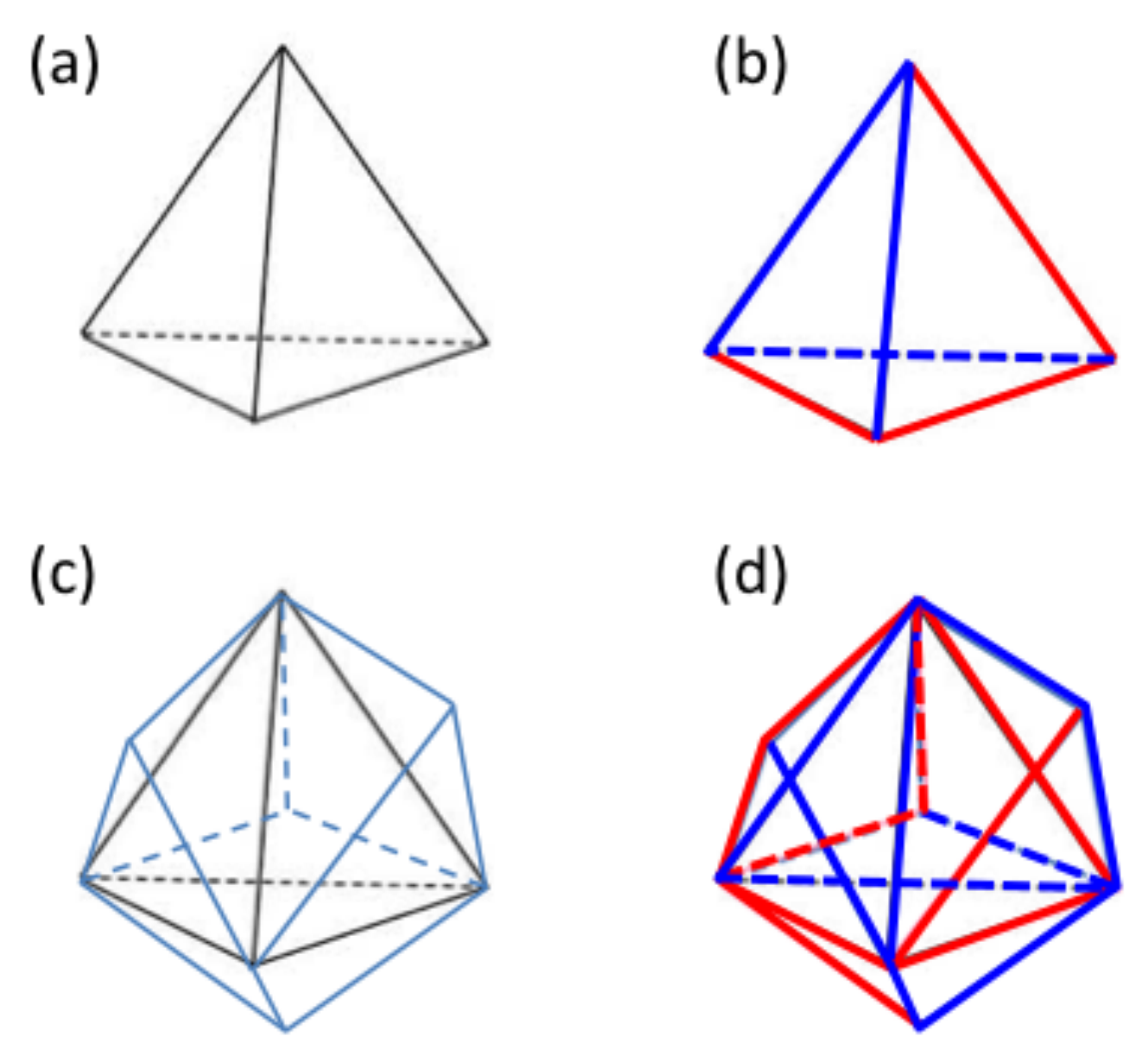}
  \caption{\label{fig:tetrahedron_partition} (color online) (a) A tetrahedron; (b) bipartition of edges in a tetrahedron (red/blue); (c)
  the same topology but with one tetrahedron grown onto each face; (d) bipartition (red/blue) of the edges in (c).} 
  \end{figure}

In addition to the topology of a torus, we also consider that of a
sphere.  The simplest geometry is a tetrahedron, shown in
Fig.~\ref{fig:tetrahedron_partition}a. The toric code ground state
on this topology is unique. 

Using the same technique of upper and
lower bounding the maximal overlap, one can easily show that its GE of spins (i.e. one-site blocks) is given by $E_G=n_s-1$,
where $n_s=4$ is the number of vertices and is related to the number
of spins $n$ via $n= d\cdot n_s/2$, where $d=3$ is the degree of
a vertex. The bipartition that gives rise to this result is shown in
Fig.~\ref{fig:tetrahedron_partition}b. 

In order to obtain a scaling we add one tetrahedron to each face, which can be iterated in order to add more sites to the system. The
first step of such a procedure is shown in
Fig.~\ref{fig:tetrahedron_partition}c. In general it is easy to obtain
$E_G=n_s-1 = 2n/d -1$, where $n_s$, $n$ and $d$ are,
respectively, the number of vertices, the number of spins (edges), and the
vertex degree of the associated geometry. This result shows that the -1 contribution is not only present in the torus topology, but  in the sphere topology as well. 
\section{Beyond the toric code}
\label{sec:beyond}

After our detailed study on the toric code model, we now turn to study the GE in other topological models that also correspond to fixed points of RG. From now on, some parts of our study will not be as detailed as for the toric code, and will focus mainly on some of the fundamental properties. Yet, in order to make everything as clear as possible, we need to see how some of the main properties of the toric code apply also to other models. In particular, there are three main points to consider:

\vspace{10pt}

\underline{\emph{1. Excitations and other ground states.-}} As for the toric code, excited states and other ground states that can be obtained by local operations acting on some reference ground state will share the same entanglement properties (c.f. Lemma \ref{lem:tc_spins}). Therefore, in many cases we shall only discuss the entanglement of some \enquote{reference} ground state.

\vspace{10pt}

\underline{\emph{2. Disentangling general stabiliser states.-}} Very importantly, the disentangling property is not exclusive of the toric code only, but it holds for any stabiliser state as well. More concretely, \emph{any state that is stabilised by a set of operators can be built from a quantum circuit acting on some separable initial state, where the gates of the quantum circuit involve the basis of the stabiliser group} (see e.g. \cite{Preskill}). If this basis of stabilisers can be thought of as acting locally on a lattice, then the quantum gates of the corresponding quantum circuit will also be local in the lattice. Moreover,  there will be a set of unitary gates in the quantum circuit that will correspond to each one of the elements of the group basis, and which will take into account the existence of that particular basis element. That is: removing the action of the gates corresponding to a particular basis element effectively removes that stabiliser from the basis of the stabiliser group (e.g. removing a plaquette or star operator from the toric code).

Therefore, simply by \emph{reversing} the action of such a quantum circuit on the stabilised state, we will always manage to \emph{disentangle} the state by means of sufficiently local unitary operations. At every step, the remaining quantum state will be the state stabilised by the remaining set of stabiliser operators. In particular, this implies that there is always a MERA representation for any stabiliser state (not necessarily topological), and the toric code model is just a particular case.

This property turns out to be quite important, since it allows us to extend straightforwardly the results for the GE of blocks beyond the toric code: when blocking is considered, simply perform disentangling operations inside of each block in order to remove the maximum of stabilisers within each block, so that the remaining stabilisers will always be between spins in different blocks. By construction, the remaining state will usually obey an area law for the GE of blocks, plus a possible topological term. Because of this, from now on the GE of blocks will not be discussed in detail except for a couple of examples.

\vspace{10pt}

\underline{\emph{3. The non-Abelian case.-}} It is indeed possible to generalise our results to systems with non-Abelian symmetries, such as quantum double models. However, in order to deal with these models, it will be worthwhile to develop part of our formalism in the most general fashion. This will be done in Sec.~\ref{sec:general_bounds}, and quantum double models will be considered subsequently in Sec.~\ref{sec:quantumdouble}.

\section{Double semion model}
\label{sec:semion}
The double semion model is given by the spin model on the
honeycomb lattice
\begin{equation}
H_{\rm DS}=-\sum_s A_s -\sum_p B'_p,
\end{equation}
where $A_s$ and $B'_p$ are mutually commuting and given by
\beq
A_s \equiv \prod_{j \in s} \sigma_x^{[j]} \ \ \ \ \ \ \ \  B'_p \equiv - \prod_{k \in {\rm legs  \ of}\ p} i^{(1-\sigma_x^{[k]})/2} \prod_{j \in p} \sigma_z^{[j]}.
\eeq
This is, the star operators are the same as for the toric code, whereas the plaquette operators are related via
 \begin{equation}
 B'_p=-B_{p} \prod_{k \in {\rm legs  \ of}\ p} i^{(1-\sigma_x^{[k]})/2},
\end{equation}
with $B_p$ the toric code plaquette operator. As for the toric code, star and plaquette operators satisfy the non-local constraint
\beq
\prod_s A_s = \prod_p B'_p = \mathbb{I}.
\eeq
This model is known to be also topologically ordered, yet its topological phase differs from the one of the toric code model. More precisely, while the toric code corresponds to the topological phase of a $\mathbb{Z}_2$ gauge theory, the double semion model corresponds to a $U(1) \times U(1)$ Chern-Simons theory \cite{sn}.

One of the ground states of this double semion model is almost identical to the toric code ground
state $\ket{+,+}$, except that each term in the decomposition
Eq.~(\ref{eqn:++GS}) is multiplied by a factor $(-1)^{X_c}$, where
$X_c$ counts the number of loops formed by spins in the $|-\rangle$ state. More specifically, this ground state
is given by
\begin{equation}
\label{eqn:++DBGS}
 |+,+\rangle_{\rm DS}\equiv \frac{1}{\sqrt{|\mathcal{G}_p|}}\sum_{g\in
\mathcal{G}_p} (-1)^{X_c} g\ket{+}^{\otimes n},
\end{equation}
where again $\mathcal{G}_p$ is the group of all possible products of plaquette operators $B'_p$. Of course, the state $\ket{+,+}$ is by construction a stabilised state of the group $\mathcal{G}_p$. We provide now the following Theorem:

\begin{theorem}
For the double semion model in a honeycomb lattice $\Sigma$, the GE
of spins for the ground state $\ket{+,+}_{\rm DS}$ is given by \beq
{E_G} = n_p-1 , \label{bo_semion} \eeq where $n_p$ is the number of
plaquettes (faces) in $\Sigma$.
\end{theorem}

\begin{proof}
Remember the ground state $\ket{+,+}$ for the toric code in the honeycomb lattice. This state can always be written as $\ket{+,+} = \sum c_{ijk...} \ket{ijk...}$ with $c_{ijk...} \ge 0$ in the local basis $\{ \ket{+}, \ket{-} \}$ for every spin. From here, it follows that $|\sum_{ijk..} p_i q_j ...  c_{ijk..} | \le  \sum_{ijk...} |p_i q_j... c_{ijk}|$, where the equality holds if $p_i, q_j,...$ are all non-negative up to a common phase, and the maximum is achieved by choosing $p_i,q_j,..$ to be non-negative. Next, let us define $d_{ijk...} = c_{ijk..} e^{i \phi}$, where $\phi$ is an arbitrary phase. Then it follows that $|\sum_{ijk..} p_i q_j ... d_{ijk..}| \le  \sum_{ijk..} |p_i q_j .... c_{ijk}|$. A maximisation over the left hand side gives the maximal overlap, say, for the double semion model by choosing appropriate phase factors and coefficients. But in any case, this will be smaller than maximising the right hand side whose maximum is, in this particular case, the maximal overlap for the toric code. From here the result just follows.
\end{proof}

From the above Theorem it follows that the topological contribution is $E_{\gamma} = 1$, which is again equal to the topological entropy for this model. Moreover, even if we do not discuss it here in detail, we expect the GE of blocks to obey similar laws as for the toric code model, since the ground state can also be disentangled by local moves \cite{lwdis}. Similarly, the entanglement properties of other ground states and excited states obtained by acting locally on $\ket{+,+}$, remain identical to those of $\ket{+,+}$ (c.f. Lemma \ref{lem:tc_spins}).

\section{Color code models}
\label{sec:color}

Topological color codes \cite{Bombin} are defined as follows: qubits reside on \emph{vertices} (not links!) of
a lattice that is a \emph{2-colex}. A 2-colex $\Sigma$ is a 2D lattice embedded in a torus of arbitrary genus $\mathfrak{g}$, with the following properties: (a) every vertex of the lattice has degree 3, i.e. precisely 3 edges meet at each vertex; (b) the faces of the lattice are 3-colorable, i.e. the faces can be coloured with 3-colours (e.g. red, green and blue) in such a way that no two adjacent faces have the same color, see e.g.
Fig~\ref{fig:ColorCode}. Two types of commuting operators are
defined on each plaquette (or face) $p$,
\begin{equation}
B_{p}^X \equiv \prod_{j \in {p}} \sigma_x^{[j]} \ \ \ \ \ \ \ \ \ B_{p}^Z \equiv \prod_{j\in {p}} \sigma_z^{[j]}.
\end{equation}
The Hamiltonian of the color code is thus
\beq
H_{{\rm CC}}=-\sum_{p}(B_{p}^X+B_{p}^Z).
\eeq
For this model in 2D lattices, the ground space has degeneracy $2^{4\mathfrak{g}}$, where $\mathfrak{g}$ is the genus of the underlying Riemann surface. For instance, on the torus this degeneracy is $16$.

Let us define the group $\mathcal{G}_p^X$ as the group generated by the products of all possible plaquette operators $B_p^X$ for all three colours. Notice that the product of all plaquette operators $B_{p}^X$ corresponding to the same color $c$ corresponds to the same element of the group $\mathcal{G}_p^X$ (namely, the action of a $\sigma_x$ operator on all the qubits). Therefore, we have the non-local constraint
\beq
\prod_{p\in red} B_{p}^X=\prod_{p\in green} B_{p}^X=\prod _{p\in blue} B_{p}^X = \prod_{v \in \Sigma} \sigma_x^{[v]}.
\eeq
The above relation implies that the size of the group is $|\mathcal{G}_p^X| = 2^{n_p-2}$, with $n_p$ the number of plaquettes in lattice $\Sigma$.

The ground level subspace of the color code model is constructed in
a very similar way as we did for the toric code. Specifically, for
the color code on a torus the ground space reads $\mathcal{L} = {\rm
span} \{ \ket{i,j,k,l}, ~ i,j,k,l=0,1 \}$, where \beq \ket{i,j,k,l}
= \frac{1}{\sqrt{|\mathcal{G}_p^X|}} \sum_{g \in \mathcal{G}_p^X} g
~ w_1^i ~ w_2^j ~ w_3^i ~ w_4^j ~ \ket{0}^{\otimes n}.
\label{gs_color} \eeq Operators $\{ w_1, w_2, w_3, w_4 \}$ above
correspond to non-contractible loop operators, see
Fig.~\ref{fig:ColorCode}.b. We do not enter here on the specific
action of these operators on every spin in the loop, and just
mention that it is local. The $16$ vectors $\ket{i,j,k,l}$ are
orthonormal and stabilised by $\mathcal{G}_p^X$, and therefore form
a possible basis of the ground level space of the model on a torus.

\begin{figure}
 \includegraphics[width=4.5cm]{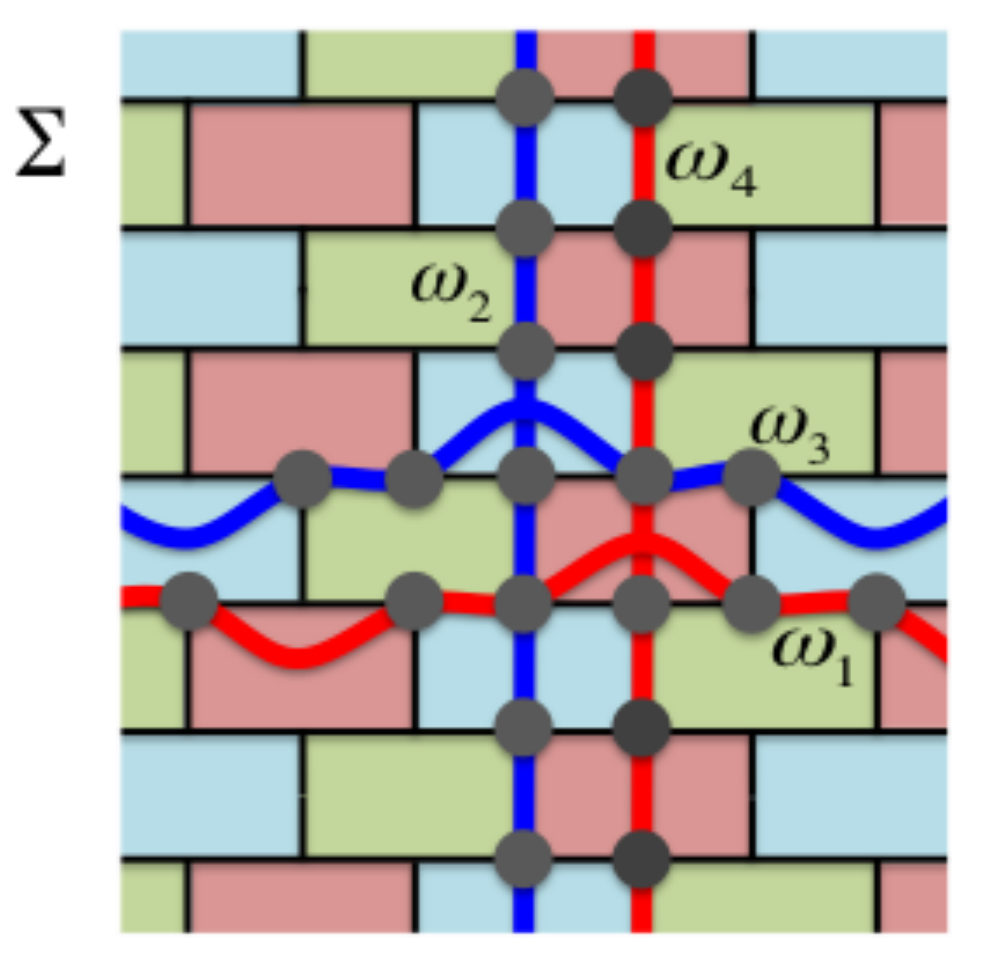}
  \caption{\label{fig:ColorCode} (color online) Color code on a honeycomb (brickwall) lattice $\Sigma$, which is an example of a 2-colex. There is a spin on every vertex, and every plaquette is coloured with either red, green or blue. Non-local string operators $w_1,w_2,w_3$ and $w_4$ are also shown, where periodic boundary conditions on a torus are assumed. Grey dots indicate the sites where string operators are acting.}
\end{figure}

As for the toric code, the action of loop operators corresponds to local unitary operations on a reference state $\ket{0,0,0,0}$, and therefore do not change the entanglement properties (c.f. Lemma \ref{lem:tc_spins}). Therefore, from now on we shall consider only
 the entanglement properties of the ground state $\ket{0,0,0,0}$, defined as
\begin{equation}
|0,0,0,0\rangle= \frac{1}{\sqrt{|\mathcal{G}_p^X|}}\sum_{g\in \mathcal{G}_p^X} g\ket{0}^{\otimes n}.
\label{ccstate}
\end{equation}

In what follows we show two different approaches to compute the GE for this model. The first approach is based, as all the calculations before, on finding upper and lower bounds for the GE. This approach works well for e.g. finding the GE of spin and blocks on regular lattices such as the honeycomb lattice. The second approach uses a different technique, building on recent work \cite{Va13}. The main idea is to exploit that color code ground states are instances of so-called Calderbank-Shor-Steane states of self-orthogonal type. This refers to the fact that $|0, 0, 0, 0\rangle$, when expanded in the standard basis, has the form \beq |0, 0, 0, 0\rangle = \frac{1}{\sqrt{|S|}} \sum_{x\in S}|x\rangle\eeq where $S\subseteq\mathbb{Z}_2^n$ is a collection of bit strings with a particular structure i.e. it is a self-orthogonal classical linear code. This property is in fact displayed by color code states associated with arbitrary lattices (2-colexes) and for surfaces of arbitrary genus. As a result, the second approach will lead to a general calculation of the GE for arbitrary color codes.

\subsection{Bound method}

Let us consider the color code on a honeycomb lattice $\Sigma$. For this lattice, and with this method, we can provide the following Theorem:

\begin{theorem}
For the color code model in a honeycomb lattice $\Sigma$, the GE of
spins for the ground state $\ket{0,0,0,0}$ is given by \beq {E_G} =
n_p-2 , \label{bo_color} \eeq where $n_p$ is the number of
plaquettes (faces) in $\Sigma$. \label{bobo}
\end{theorem}

\begin{proof}
First, since the group $\mathcal{G}_p^X$ has size $|\mathcal{G}_p^X|=2^{n_p-2}$, we obtain a lower
bound on the maximal overlap $\Lambda_{\max}\ge |\mathcal{G}_p^X|^{-1/2}$, and
hence an upper bound on the GE given by $E_G \le n_p-2$. Second, in order to
derive a good lower bound on the GE, we consider the natural
bipartition for the honeycomb lattice, see
Fig.~\ref{fig:ColorCode2}.a. One can convince oneself easily that any group
element in $\mathcal{G}_p^X$ must have nontrivial support contained in both the $A$ and
$B$ sets, and therefore the subgroups $\mathcal{G}_p^X(A)$ and $\mathcal{G}_p^X(B)$
are both trivial and contain only the identity element. For the color code, the reduced density matrix of $A$ also satisfies the property $\rho_A^2 = (|\mathcal{G}_p^X(A)||\mathcal{G}_p^X(B)|/|\mathcal{G}_p^X|) \rho_A$ \cite{Kargarian}. Therefore,
$\Lambda_{\max}^2\le 1/|\mathcal{G}_p^X|$, and hence $E_G \ge n_p -2$. We therefore
arrive at $E_G=n_p-2$, which proves the theorem.
\end{proof}

\begin{figure}
 \includegraphics[width=9cm]{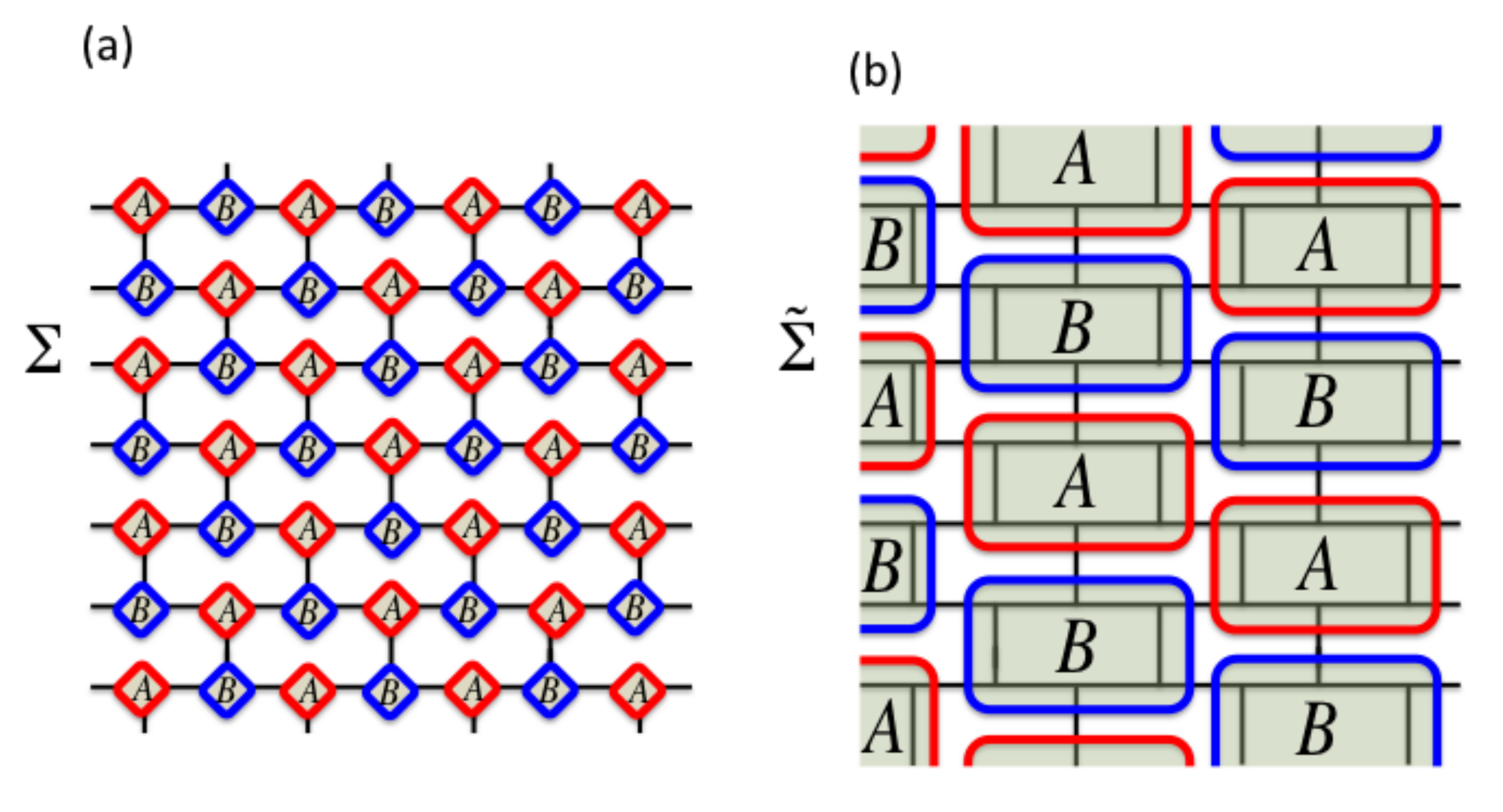}
  \caption{\label{fig:ColorCode2} (color online) (a) Natural bipartition of the honeycomb lattice $\Sigma$ into $A$ and $B$ sets (respectively red and blue). (b) Possible blocking and bipartition of lattice $\widetilde{\Sigma}$, obtained after applying disentanglers corresponding to the stabilisers within each block to the honeycomb lattice $\Sigma$ of the color code. We also show the relevant bipartition into $A$ and $B$ sets (again red and blue).}
\end{figure}

Some comments are in order. First, we see that the term of topological origin is $E_\gamma = 2$, which is once again identical to the topological entanglement entropy \cite{Kargarian}. Second, this topological geometric entanglement is
twice that of the toric code, which is a direct consequence that the color
code model we consider here is equivalent to two copies of the toric code under some given RG procedure \cite{Bombin2}. Third, a similar result holds also for the ground state associate with the group $\mathcal{G}_p^Z$ generated by the product of $B_p^Z$ operators. Fourth, excitations and other ground states related by local operations will have the same entanglement. And fifth, the GE of blocks will consist again of an area-law term plus a topological contribution which will be again $E_\gamma = 2$. This can be seen, e.g. by choosing the blocking from Fig.~\ref{fig:ColorCode2}.b, and appropriately disentangling qubits within each block.

\subsection{Color codes and self-orthogonal classical codes}

Here we provide a generalisation of theorem \ref{bobo} which will allow to compute the GE of color code states in a much broader setting i.e. for arbitrary lattices (2-colexes) and arbitrary genus. The technique used to arrive at this generalisation is entirely different from the proof of theorem \ref{bobo}. It involves connections between color codes and classical coding theory which were found in the recent work \cite{Va13}.

We consider a color code defined on an arbitrary 2-colex embedded in a surface of arbitrary genus. In analogy with (\ref{ccstate}) we consider the state \beq\label{TCC_state} |\mathbf{0}\rangle = \frac{1}{\sqrt{|{\cal G}_p^X|}} \sum_{g\in {\cal G}_p^X} g|0\rangle^{\otimes n}.\eeq
We will show here that the result in theorem \ref{bobo} extends to all states $|\mathbf{0}\rangle$:

\begin{theorem}\label{thm_GE_color_general}
Consider the color code ground state $|\mathbf{0}\rangle$ for an arbitrary 2-colex $\Sigma$ embedded in a surface of arbitrary genus. The GE of spins for this state is \beq E_G = n_p -2.\eeq
\end{theorem}

Before proving the result, we develop some preliminary material. A linear subspace $C\subseteq \mathbb{Z}_2^n$ is called a (classical) binary linear code of length $n$. The elements of $C$ are called its codewords. For every binary linear code $C$ of length $n$, define an $n$-qubit state \beq|C\rangle:=\frac{1}{\sqrt{|C|}}\sum_{u\in C}|u\rangle.\eeq We shall refer to any state of this kind as a Calderbank-Shor-Steane (CSS) state. It easily follows from (\ref{TCC_state}) that $|\mathbf{0}\rangle$ is a CSS state. More precisely, we have  $|\mathbf{0}\rangle \equiv |C_p^X\rangle$  where the linear code $C_p^X$ is defined as follows: writing every element $g\in {\cal G}_p^X$ as \beq g= X^{u_1}\otimes\cdots\otimes X^{u_n}\eeq with $u_i \in \mathbb{Z}_2$, the code $C_p^X$ is simply given by the collection of all bit strings $(u_1, \dots, u_n)$ arising from elements of ${\cal G}_p^X$ in this way. In particular, one has \beq\label{order_C_p^X} |C_p^X| = |{\cal G}_p^X| = 2^{n_p-2},\eeq which will be important below.

It turns out that color code states are CSS states of a particular kind. A linear code $C$ is called self-orthogonal if every two codewords are orthogonal: $s^Tt=0$ for every $s, t\in C$. A CSS state for which the underlying classical code is self-orthogonal will be denoted \emph{CSS state of self-orthogonal type}. In \cite{Va13} the following was shown:

\begin{lemma}\label{thm_TCC_selforthogonal}
Consider a color code associated with an arbitrary 2-colex $\Sigma$. Then the ground state $|\mathbf{0}\rangle = |C_p^X\rangle$ is a CSS state of self-orthogonal type.
\end{lemma}
Next we provide a calculation of the GE of arbitrary CSS states of self-orthogonal type---by lemma \ref{thm_TCC_selforthogonal} this will yield the GE of color code ground states. Note however that not all CSS states of self-orthogonal type are color code ground states i.e. our results extend beyond the color code setting. It is also interesting to point out that toric code ground states, even though they are CSS states, are not of self-orthogonal type.

Consider an $n$-qubit system and let $|\Phi\rangle=\bigotimes |\phi^{[i]}\rangle$ be an $n$-qubit product state, where $|\phi^{[i]}\rangle = a_i|0\rangle + b_i|1\rangle$. Define an associated $n$-qubit tensor product operator as follows:
 \beq\label{A} A:= \bigotimes_i A_i, \quad \mbox{where} \quad A_i := \left[ \begin{array}{cc}a_i & b_i \\ b_i & -a_i\end{array} \right].\eeq
Now consider an arbitrary self-orthogonal linear code $C$ of length $n$ and the associated $n$-qubit CSS state $|C\rangle$. The following lemma, proved in \cite{Va13},  relates the overlap $\langle C| \Phi\rangle$ to the expectation value $\langle C| A|C\rangle$:

\begin{lemma}\label{thm_overlap_expectation}
Let $C$ be a binary self-orthogonal linear code of length $n$ and let $|\Phi\rangle$ be an $n$-qubit product state. Then \beq \sqrt{|C|} \cdot \langle C| \Phi\rangle = \langle C| A|C\rangle.\eeq
\end{lemma}

Lemma \ref{thm_overlap_expectation} will now be used to prove:

\begin{theorem}\label{thm_E_g_selforthogonal}
Let $|C\rangle$ be an arbitrary $n$-qubit CSS state of self-orthogonal type. Then the GE of spins is \beq E_G(|C\rangle) = \log |C|.\eeq
\end{theorem}
\begin{proof} As before, let $\Lambda_{\mbox{\scriptsize{max}}}$ denote the maximal value of $|\langle C|\Phi\rangle|$ when $|\Phi\rangle$ ranges over all product states. First, taking an arbitrary $u\in C$ yields $|\langle C|u\rangle|= |C|^{-\frac{1}{2}}$ so that we obtain the lower bound $\Lambda_{\mbox{\scriptsize{max}}}\geq |C|^{-\frac{1}{2}}$. Second, since $|C\rangle$ has real and nonnegative amplitudes in the standard basis, the maximal value of $|\langle C|\Phi\rangle|$  will be reached by a product state which has real,  nonnegative amplitudes as well. Consider an arbitrary product state $|\Phi\rangle$ with real and nonnegative amplitudes and the associated tensor product operator $A$ as in (\ref{A}). Using lemma \ref{thm_overlap_expectation} we find
\begin{eqnarray} |\langle C| \Phi\rangle| &=& \frac{1}{\sqrt{|C|}}|\langle C| A|C\rangle| \leq \frac{1}{\sqrt{|C|}}\|A\| \nonumber\\ &=& \frac{1}{\sqrt{|C|}}\prod_i\|A_i\|
\end{eqnarray}
 where $\|\cdot\|$ denotes the operator norm. Since each  $|\phi^{[i]}\rangle$ is a unit vector with real and nonnegative amplitudes, each $A_i$ is a real orthogonal matrix. Hence $\|A_i\|= 1$. This proves the upper bound $\Lambda_{\mbox{\scriptsize{max}}}\leq |C|^{-\frac{1}{2}}$ which matches the previously found lower bound. The result now follows immediately. \end{proof}

Combining lemma \ref{thm_TCC_selforthogonal}, theorem \ref{thm_E_g_selforthogonal} and Eq. (\ref{order_C_p^X}) immediately proves theorem \ref{thm_GE_color_general}.

\section{General Formalism for the Bound Method}
\label{sec:general_bounds}

In order to extend our results to the non-Abelian setting (e.g. quantum double models), it is worthwhile to generalise the bounding technique that we have used so far. Looking back, we have calculated the GE of quantum states with a certain (Abelian) group symmetry, and our bounding strategy involved two steps. First, we chose suitable product states to lower bound the maximal over- lap (and hence to upper bound the GE). Second, we picked certain bipartitions and determined the maximal overlap with bipartite product states in order to upper bound the maximal overlap with multipartite product states (and hence to lower bound the GE). As we have seen, this procedure gives \emph{tight} bounds in almost all the cases considered.

We shall now distill the essence of this bounding strategy and derive general bounds on the GE of quantum states with a non-Abelian group symmetry. Other than placing all of our previous examples under a common roof these bounds may even be of independent interest beyond the present study of topological order.

\subsection{Overview}

Let $G$ be a finite group. Suppose its $n$-fold direct product
$\mathcal{G}=G^{\times n}$ acts on some $m$-partite Hilbert
space~$\mathcal{H}=\mathcal{H}_1\otimes\dots\otimes\mathcal{H}_m$
via product operators~$O^t=O_1^t\otimes\dots\otimes O_m^t$.
Consider states of the form
\begin{equation}
    \label{eq:general_state}
    \sum_{t\in
          \mathcal{G}}
    O^t\mkern2mu
    \ket{\Phi}
\end{equation}
where $\ket{\Phi}=\ket{\phi^{[1]}}\otimes\dots\otimes\ket{\phi^{[m]}}$
is some reference product state such that
any two local states~$O_k^t\mkern2mu\ket{\phi^{[k]}}$
and~$O_k^u\mkern2mu\ket{\phi^{[k]}}$ are either identical or
orthogonal.

We can define the \emph{global stabilizer}~$\mathcal{G}_\Phi\subset\mathcal{G}$
as the subgroup which leaves the reference state~$\ket{\Phi}$ invariant.
Note that $\mathcal{G}_\Phi$ need
\emph{not} be normal in general, so the collection
$\mathcal{G}'\coloneqq\mathcal{G}/\mathcal{G}_\Phi$ of cosets is not necessarily
a group. It is clear that the state
\begin{equation}
    \label{eq:state}
    \ket{\Psi}
    =\frac{1}
          {\sqrt{\mathcal{G}'}}
     \sum_{t\in
           \mathcal{G}'}
     O^t\mkern2mu
     \ket{\Phi}
\end{equation}
is correctly normalized. Picking any component~$O^t\mkern2mu\ket{\Phi}$
we get a trivial lower bound
on the maximal overlap~$\Lambda_\mathrm{max}(\Psi)$ and hence
by~\eqref{eq:Entrelate} an upper bound
$E_G(\Psi)\le n\log_2\abs{G}-\log_2\abs{\mathcal{G}_\Phi}$.

Now let $\mathcal{H}=\mathcal{H}_A\otimes\mathcal{H}_B$ a bipartition
of the Hilbert space into subsystems~$A$ and~$B$
with the reference state $\ket{\Phi}=\ket{\phi_A}\otimes\ket{\phi_B}$
partitioned accordingly.
For each operator~$O^t$ we define a truncated
operator~$O_X^t$ as the restriction of~$O^t$ to~$\mathcal{H}_X$
and the identity in~$\mathcal{H}_{\bar{X}}$.
We furthermore define the \emph{local stabiliser}
$\mathcal{G}_\Phi(X)\subset\mathcal{G}$ for subsystem~$X$ as the
subgroup whose truncated operators~$O_X^t$ leave the reference
state~$\ket{\Phi}$ invariant. It is clear that the global stabiliser is
contained in any local stabiliser, so
$\mathcal{G}_\Phi'(X)\coloneqq\mathcal{G}_\Phi(X)/\mathcal{G}_\Phi$
makes sense.
Note that again $\mathcal{G}_\Phi'(X)$ is not necessarily a group.
Nevertheless, for bipartitions with the property
\begin{equation}
    \label{eq:relation_local_stabilizers}
    \mathcal{G}_\Phi(A)
    \supseteq
     \mathcal{G}_\Phi(B)
\end{equation}
we will find (see Lemma~\ref{lem:spectrum} below) that
the largest eigenvalue of the reduced density operator~$\rho_A$
equals~$\abs{\mathcal{G}_\Phi(A)}\mkern2mu\abs{\mathcal{G}_\Phi(B)}/\bigl(\abs{\mathcal{G}}\mkern2mu\abs{\mathcal{G}_\Phi}\bigr)$
which by~\eqref{eq:single-copy} immediately implies a lower bound on
the \textsc{GE}. We can summarize both bounds on the \textsc{GE}
in the following

\begin{theorem}[General bound]
    Let $\ket{\Psi}$ be a state of the form~\eqref{eq:general_state}.
    If there is a bipartition $\mathcal{H}=\mathcal{H}_A\otimes\mathcal{H}_B$
    obeying~\eqref{eq:relation_local_stabilizers}, the geometric
    entanglement of~$\ket{\Psi}$ is bounded by
    \begin{equation}
        -\log_2
        \frac{\abs{\mathcal{G}_\Phi(A)}\mkern2mu
              \abs{\mathcal{G}_\Phi(B)}}
             {\abs{\mathcal{G}_\Phi}}
        \le
         E_G(\Psi)-
         n
         \log_2
         \abs{G}
        \le
         -\log_2
         \abs{\mathcal{G}_\Phi}.
    \end{equation}
\end{theorem}

As is evident the derivation of this general bound is
\emph{independent} of the reference state~$\ket{\Phi}$ or
any physical model in which the state~$\ket{\Psi}$ might arise,
including any geometry or topology that might be involved.
So the problem of obtaining actual bounds for specific states in specific models
with specific geometries (topologies) reduces to the much simpler problem
of analysing the interplay between the group~$\mathcal{G}$,
the global stabiliser~$\mathcal{G}_\Phi$
and the local stabilisers~$\mathcal{G}_\Phi(X)$ in the particular case at hand. Additionally,
specific models will often allow us to characterise these stabilisers
in purely geometric (combinatorial) terms, which provides a constructive
way to find bipartitions obeying~\eqref{eq:relation_local_stabilizers}.
Remarkably, this strategy even
yields \emph{exact} values for the geometric entanglement in many
interesting cases as we will show (and have seen already).
Namely, with a mild assumption in addition to~\eqref{eq:relation_local_stabilizers}
we immediately get

\begin{theorem}
    \label{thm:general_value}
    Let $\ket{\Psi}$ a state of the form~\eqref{eq:general_state}.
    If there is a bipartition $\mathcal{H}=\mathcal{H}_A\otimes\mathcal{H}_B$ such that
    \begin{equation}
        \mathcal{G}_\Phi(A)
        =\mathcal{G}_\Phi(B)
        =\mathcal{G}_\Phi,
    \end{equation}
    the geometric entanglement of~$\ket{\Psi}$ is given by
    \begin{equation}
        E_G(\Psi)
        =n
         \log_2
         \abs{G}-
         \log_2
         \abs{\mathcal{G}_\Phi}.
    \end{equation}
\end{theorem}

We would like to mention that if $G$ is Abelian then both~$\mathcal{G}'$ and~$\mathcal{G}_\Phi'(X)$
are indeed groups and one can determine the spectrum of~$\rho_A$ directly
in terms of these. This has been discussed in Ref.~\cite{Hamma}
for the special case~$G=\mathbb{Z}_2$. There our~$\mathcal{G}'$ is denoted by~$G$ and our~$\mathcal{G}_\Phi'(A)$ by~$G_A$
(and similarly for the other subsystem~$B$).

\subsection{Some Details}

Here we will supply some additional details glossed over before.
First the precise definitions of global and local stabilisers.

\begin{definition}[Global stabiliser]
    \begin{equation}
        \mathcal{G}_\Phi
        \coloneqq
         \{t\in
           \mathcal{G}\mid
           O^t\mkern2mu
           \ket{\Phi}
           =\ket{\Phi}\}.
    \end{equation}
\end{definition}

\begin{definition}[Local stabiliser]
    \begin{equation}
        \mathcal{G}_\Phi(X)
        \coloneqq
         \{t\in
           \mathcal{G}\mid
           O_X^t\mkern2mu
           \ket{\Phi}
           =\ket{\Phi}\}.
    \end{equation}
\end{definition}

This shows how we obtain arbitrary reduced density operators
from the state~\eqref{eq:state}.
\begin{lemma}[Reduced density operator]
    \label{lem:rho}
    For any bipartition~$\mathcal{H}=\mathcal{H}_A\otimes\mathcal{H}_B$
    into subsystems~$A$ and~$B$ one has
    \begin{equation}
        \rho_A
        =\frac{1}
              {\abs{\mathcal{G}}\mkern2mu
               \abs{\mathcal{G}_\Phi}}
         \sum_{u\in
               \mathcal{G}}
         \sum_{v\in
               \mathcal{G}_\Phi(B)}
         O_A^u\mkern2mu
         \ketbra{\phi_A}
                {\phi_A}\mkern2mu
         O_A^{v
              u^{-1}}.
    \end{equation}
\end{lemma}
\begin{proof}
    Avoiding the sum over cosets, we may
    rewrite~\eqref{eq:state} as
    \begin{equation*}
        \ket{\Psi}
        =\frac{1}
              {\sqrt{\abs{\mathcal{G}}\mkern2mu
                     \abs{\mathcal{G}_\Phi}}}
         \sum_{t\in
               \mathcal{G}}
         O_A^t\mkern2mu
         \ket{\phi_A}\otimes
         O_B^t\mkern2mu
         \ket{\phi_B},
    \end{equation*}
    hence the global density operator reads
    \begin{equation*}
        \rho
        =\frac{1}
              {\abs{\mathcal{G}}\mkern2mu
               \abs{\mathcal{G}_\Phi}}
         \sum_{u,
               v\in
               \mathcal{G}}
         O_A^u\mkern2mu
         \ketbra{\phi_A}
                {\phi_A}\mkern2mu
         O_A^{v^{-1}}\otimes
         O_B^u\mkern2mu
         \ketbra{\phi_B}
                {\phi_B}\mkern2mu
         O_B^{v^{-1}}.
    \end{equation*}
    By assumption $\{O_B^t\mkern2mu\ket{\phi_B}
    \mid t\in\mathcal{G}/\mathcal{G}_\Phi(B)\}$
    is a set of mutually orthogonal states in subsystem~$B$ so we
    can use it to take the partial trace:
    \begin{equation*}
        \tr_B(\rho)
        =\frac{1}
              {\abs{\mathcal{G}}\mkern2mu
               \abs{\mathcal{G}_\Phi}}
         \sum_{u,
               v\in
               \mathcal{G}}
         O_A^u\mkern2mu
         \ketbra{\phi_A}
                {\phi_A}\mkern2mu
         O_A^{v
              u^{-1}}\mkern-4mu\cdot
         \expval{\phi_B}
                {O_B^v}
                {\phi_B}.
    \end{equation*}
    Since the local expectation value is non-zero precisely for
    $v\in\mathcal{G}_\Phi(B)$ we obtain
    \begin{equation*}
        \tr_B(\rho)
        =\frac{1}
              {\abs{\mathcal{G}}\mkern2mu
               \abs{\mathcal{G}_\Phi}}
         \sum_{u\in
               \mathcal{G}}
         \sum_{v\in
               \mathcal{G}_\Phi(B)}
         O_A^u\mkern2mu
         \ketbra{\phi_A}
                {\phi_A}\mkern2mu
         O_A^{v
              u^{-1}}.
    \end{equation*}
\end{proof}

We turn to those bipartitions obeying~\eqref{eq:relation_local_stabilizers}.
\begin{lemma}[Spectrum of reduced density operators]
    \label{lem:spectrum}
    Let $\mathcal{H}=\mathcal{H}_A\otimes\mathcal{H}_B$ a bipartition such that
    $\mathcal{G}_\Phi(A)\supseteq\mathcal{G}_\Phi(B)$. Then the spectrum
    of~$\rho_A$ is flat with non-zero eigenvalues
    \begin{equation}
        \frac{\abs{\mathcal{G}_\Phi(A)}\mkern2mu
              \abs{\mathcal{G}_\Phi(B)}}
             {\abs{\mathcal{G}}\mkern2mu
              \abs{\mathcal{G}_\Phi}}.
    \end{equation}
\end{lemma}
\begin{proof}
    From Lemma~\ref{lem:rho} we conclude that
    \begin{equation*}
        \rho_A
        =\frac{\abs{\mathcal{G}_\Phi(B)}}
              {\abs{\mathcal{G}}\mkern2mu
               \abs{\mathcal{G}_\Phi}}
         \sum_{t\in
               \mathcal{G}}
         O_A^t\mkern2mu
         \ketbra{\phi_A}
                {\phi_A}\mkern2mu
         O_A^{t^{-1}},
    \end{equation*}
    and now the observation
    \begin{equation*}
        \rho_A^2
        =\frac{\abs{\mathcal{G}_\Phi(A)}\mkern2mu
               \abs{\mathcal{G}_\Phi(B)}}
              {\abs{\mathcal{G}}\mkern2mu
               \abs{\mathcal{G}_\Phi}}\mkern2mu
         \rho_A
    \end{equation*}
    is enough to prove the claim.
\end{proof}

\section{Quantum double models}
\label{sec:quantumdouble}

Consider a planar graph $\Gamma=(V,E,F)$ with vertices~$V$, edges~$E$ and
faces~$F$, and sizes $\abs{V}=n_s$, $\abs{E}=n$ and $\abs{F}=n_p$. The Hilbert space of the quantum double model is defined
as~$\mathcal{H}\coloneqq\mathcal{H}_1\otimes\dots\otimes\mathcal{H}_{\abs{E}}$ where each
local~$\mathcal{H}_i\simeq\mathbb{C}G$. Vertex projectors (acting on stars) read
\begin{equation*}
    A_s
    =\frac{1}
          {\abs{G}}
     \sum_{g\in G}
     A^g_s
\end{equation*}
where each individual vertex operator $A^g_s$ acts on a vertex (star)~$s\in V$ by
multiplying all edges meeting at~$s$ by~$g\in G$ from the \emph{left}, provided all
edges point towards~$s$. Otherwise an edge is multiplied by~$g^{-1}$ from the
\emph{right}. It is clear that vertex operators on different vertices commute. Similarly, plaquette projectors are defined as
\begin{equation*}
    B_p
    =\delta\biggl(\mkern3mu
                  \prod_{\text{$i$ along $\mathcal{C}_p$}}\mkern-6mu
                  g_i,
                  e\biggr),
\end{equation*}
which select configurations where the ordered product of group elements along an oriented circuit $\mathcal{C}_p$ around $p$ is the unit element of $G$. The Hamiltonian is the sum of these mutually commuting projectors over vertices and plaquettes, namely
\begin{equation*}
    H_{\mathrm{D}(G)}
    =-\sum_{s
            \in
            V}
     A_s-
     \sum_{p
           \in
           F}
     B_p.
\end{equation*}

We identify the direct product group from our general discussion in the previous section as
$\mathcal{G}\coloneqq G^{\times\abs{V}}$ and obtain its natural
action on the Hilbert space by collecting individual vertex
operators into the \emph{joint} vertex operator
\begin{equation*}
    O^t
    \coloneqq
     \prod_{s\in V}
     A^{g_s}_s
\end{equation*}
where $t=(g_1,\dots,g_{\abs{V}})\in\mathcal{G}$. Clearly,
any joint vertex operator has product form: given a
directed edge~$(s,s')$ with value~$\ket{x}$ an element $t=(\dots,g_s,\dots,
g_{s'},\dots)\in\mathcal{G}$ is easily seen to act locally as
\begin{equation}
    \label{eq:local_action}
    O_{(s,
        s')}^t\mkern2mu
    \ket{x}
    =\ket{g_{s'}
          x
          g_s^{-1}},
\end{equation}
thus $O^t=O_1^t\otimes\dots\otimes O_{\abs{E}}^t$.

We can obtain a ground state of the quantum double model by
choosing the reference product state
$\ket{e}\coloneqq\bigotimes_E\ket{e}$ and projecting it on the
common $+1$~eigenspace of all vertex operators
\begin{equation}
    \label{eq:gs_raw}
    \ket{\Omega} =
    \prod_{s\in V}
    A_s\mkern2mu
    \ket{e}
    =\frac{1}
          {\abs{\mathcal{G}}}
     \sum_{t\in\mathcal{G}}
     O^t\mkern2mu
     \ket{e}.
\end{equation}

In order to proceed we need to figure out the actual form of
the global and local stabilisers in the quantum double model
for our particular reference state. As for elements of the global
stabiliser~$\mathcal{G}_e$, it
follows immediately from~\eqref{eq:local_action} that $g_{s'}=g_s$
iff vertices~$s$ and~$s'$ are connected by an edge.
In particular, we have $\mathcal{G}_e\simeq G$ if the graph~$\Gamma$ is connected.
Note that $\mathcal{G}_e$ is \emph{not} normal in~$\mathcal{G}$ generally.
As far as the local stabilisers~$\mathcal{G}_e(X)$ are
concerned, we can equivalently define them combinatorially. Partition
the vertices~$V$ into clusters~$V_i(X)$ based on whether vertices
are connected by an edge in~$E_X$. (That is, two vertices are in the
same cluster iff they can be connected by a path which lies
completely in~$E_X$.) Then $t=(g_1,\dots,
g_{\abs{V}})\in\mathcal{G}_e(X)$ iff its components~$g_s$ are
constant on clusters. If~$E_X$ connects \emph{all} vertices into
a single cluster then $\mathcal{G}_e(X)\simeq G$.
It turns out that local and global stabilisers
are related quite favourably:

\begin{lemma}[Joint local stabiliser]
    \label{lem:joint_local_stabilizer}
    $\mathcal{G}_e(A)\cap\mathcal{G}_e(B)=\mathcal{G}_e$.
\end{lemma}
\begin{proof}
    We want to cover the vertices~$V$ with a set~$\{V_i(X)\}$ of
    \emph{overlapping} clusters obtained from any of the subsystems~$A$
    or~$B$. Then if $t=(g_1,\dots,g_{\abs{V}})\in\mathcal{G}_e(A)\cap
    \mathcal{G}_e(B)$ its components~$g_s$ must be constant on each cluster,
    and hence constant on the whole of~$V$ because of the overlaps. Clearly,
    we can obtain a global cover by constructing a local cover for the
    vertices along each (edge) path~$\gamma$.

    So let~$\gamma$ be a path starting at vertex~$s$ and let $V_1(A)$ be the
    enveloping cluster of~$s$ as given by subsystem~$A$. If $\gamma$ never
    leaves this cluster we are done. Otherwise we may assume that $s$ lies at
    the boundary of~$V_1(A)$, hence we will reach a distinct cluster~$V_2(A)$
    by advancing a \emph{single} edge~$(s,s')$ along~$\gamma$. Clearly, this
    edge must be in~$E_B$ and thus there exists a cluster~$V_1(B)$ containing
    both~$s$ and~$s'$. Then by induction $\{V_1(A),V_1(B),V_2(A),\dots,
    V_l(A)\}$ is a local overlapping cover along~$\gamma$.
\end{proof}

So in
order to find a suitable bipartition with
$\mathcal{G}_e(A)=\mathcal{G}_e(B)=\mathcal{G}_e$ we simply need to select a
subset~$E_A$ of edges such that both~$E_A$ and~$E_B=E\setminus E_A$
have the single cluster property.
Indeed, for the square, triangular and Kagome lattice such bipartitions
exist, see Fig.~\ref{part}(a) and Fig.~\ref{fig:Triangular} for examples.
Using Theorem~\ref{thm:general_value} we then obtain

\begin{theorem}
    For any quantum double model on a square, triangular or Kagome
    lattice $\Sigma$, the ground state $\ket{\Omega}$
    has the geometric entanglement
    \begin{equation}
        E_G
        =\alpha\mkern2mu
         \abs{V}
         -\log_2
         \abs{G}.
    \end{equation}
    where $\alpha=\log_2\abs{G}$.
\end{theorem}
Note that $\abs{V}=n_s$ (number of stars) is proportional to the volume of the surface.

From~\cite{tcmera} and the proof of Theorem~\ref{thm:tc_blocks}
it becomes clear that we can calculate
the \textsc{GE} of blocks of linear size~$k\ge 2$ directly from the renormalized
graphs~$\Gamma_k$ shown in Fig.~\ref{fig:diag}(e). On a square lattice of
linear size~$2k\mkern2mu\lambda$ there are $N_k=2\lambda^2$ such blocks, and the renormalized
graph~$\Gamma_k$ has $\abs{V_k}=N_k\mkern2mu(k+1)$ vertices.
We merely state the result for the \textsc{GE} of blocks.

\begin{theorem}
    For any quantum double model on a square
    lattice $\Sigma$, the ground state $\ket{\Omega}$
    has the geometric entanglement of blocks of size~$k$
    \begin{equation}
        E_G
        =\alpha\mkern2mu
         \abs{V_k}
         -\log_2
         \abs{G},
    \end{equation}
    where $\alpha=\log_2\abs{G}$.
\end{theorem}
Clearly, $\abs{V_k}$ is proportional to \emph{both} the renormalized
volume of the surface \emph{and} the boundary area of a single block
(with a natural value of~$4k$). This is exactly the generalisation of~\eqref{getcL} to the non-Abelian case. As a final remark, we notice that the topological contribution to the geometric entanglement for quantum double models is given by $E_\gamma=\log_2\abs{G}$. This again coincides with the corresponding value of the topological entanglement entropy.

\section{$E_{\gamma}$ away from fixed RG points}
\label{sec:beyondRG}

So far we just discussed the topological contribution to the GE for specific ground states of models that turn out to be RG fixed points, and therefore representative of their respective topological phases. But, what if we are away from the RG fixed point? Is $E_\gamma$ robust under a perturbation? In this section we address briefly this question, and arrive to the conclusion that, indeed, $E_\gamma$ is a robust property of the topological phase.

There are several ways of checking the robustness of $E_{\gamma}$ under perturbations. One possibility is to perform a numerical analysis. { In this sense, it should be possible to do large-scale calculations using tensor network methods. While we shall not carry out these calculations explicitly in this paper, we will explain different strategies based on tensor network algorithms in the Appendix. Nevertheless, we provide a small-size exact calculation at the end of this section.} Another option is a perturbation theory analysis. Yet, a more intuitive alternative is the following argumentation based on RG fixed points: all the models considered here are RG fixed points and hence representatives of their respective topological phases. The long-distance properties of any state in one of these phases do not change under local RG transformations and,  thus,  are equivalent to those of the fixed point. This, in particular, is true for the long-range pattern of entanglement, and hence for $E_{\gamma}$. Thus, any non-relevant and short-range perturbation driving the Hamiltonian away from the fixed point will produce ground states with the same $E_{\gamma}$. Nevertheless, we expect a change in the short-range pattern of entanglement, and hence in the { bulk} term corresponding to a boundary law. For short-range perturbations the change involves modifications of the pre-factor of the boundary law as well as the possible appearance of sub-leading $O(L^{-\nu'})$ corrections. Hence, Eq.(\ref{toge}) applies when away from the fixed point.

To double-check the above claim, we now perform a simple perturbation theory
analysis of the robustness of the GE of the toric code model under
external short-range perturbations such as magnetic fields. This
analysis provides upper and lower bounds for the GE, and complements
the above argumentation on the robustness of $E_{\gamma}$ based on RG.

Let us then add a perturbation to the toric code Hamiltonian on the square lattice, and see how the ground state changes. For simplicity, we consider the case of an infinite plane, where the ground state $\ket{0,0}$ is non-degenerate, and hence we can use non-degenerate perturbation theory. The perturbed Hamiltonian will be

\beq
H^{\lambda} = H_{{\rm TC}} + \lambda V ,
\eeq
where $H_{{\rm TC}}$ is the toric code Hamiltonian, $V$ is the perturbation, and $\lambda \ll 1$. Non-degenerate perturbation theory says that the new ground state can be approximated as

\beq
\ket{0,0}^{\lambda} \approx \ket{0,0} + \lambda \sum_{\phi,c} \frac{\bra{\phi,c,0,0}V\ket{0,0}}{E_{0,0} - E_{\phi,c}} \ket{\phi,c,0,0} \,
\eeq
where $E_{0,0}$ is the ground state energy and $E_{\phi,c}$ is the energy of the excited state $\ket{\phi,c,0,0}$.

We now consider the case in which the perturbation is an homogeneous magnetic field, e.g. in the $x$ direction,

\beq
V = \sum_{j = 1}^n \sigma_{x}^{[j]}
\eeq
(the case of $z$ and $y$ directions can be considered similarly). It is easy to check that in this case, the normalized perturbed ground state becomes

\beq
\ket{0,0}^{\lambda} \approx C\left( \ket{0,0}  - \frac{\lambda}{\Delta} \sum_{j = 1}^n \sigma_{x}^{[j]} \ket{0,0} \right) \ ,
\eeq
with $\Delta$ the energy gap to create a pair of flux and anti-flux quasiparticles, and $C = (1+n \lambda^2/\Delta^2)^{-1/2}$ a normalisation constant.

Our aim now is to estimate the maximum overlap of the previous state with a product state of blocks of boundary size $L$. This can be done as follows: first, and as in the unperturbed case, we apply CNOT operations locally inside of the blocks so that qubits are disentangled in the unperturbed ground state. By doing this, we can focus on the entanglement of the state

\begin{widetext}

\beq
\ket{0,0}^{\lambda}_{{\rm disentangled}} \approx C\left( \ket{\widetilde{0,0}}  - \frac{\lambda}{\Delta} \sum_{j = 1}^n \sigma_{x}^{[j]} \ket{\widetilde{0,0}} \right) \otimes \ket{e_1} \otimes \cdots \otimes \ket{e_p} \ ,
\eeq
where $\ket{e_k}$ is the quantum state for the $k$-th disentangled qubit. The above equation is indeed equivalent to

\beq
\ket{0,0}^{\lambda}_{{\rm disentangled}} \approx C\left( \ket{\widetilde{0,0}} \otimes \ket{e_1} \otimes \cdots \otimes \ket{e_p}
 - \frac{\lambda}{\Delta} \sum_{j = 1}^{n_b} S_{x}^{[j]} \ket{\widetilde{0,0}} \otimes \ket{e_1} \otimes \cdots \otimes \ket{e_p} - \frac{\lambda}{\Delta}  \ket{\widetilde{0,0}}  \ket{\omega_{1,\dots, p}}\right)\ ,
\eeq
\end{widetext}
where $S_{x}^{[j]}$ is the total spin in the $x$ direction for the $L$ spins in the boundary of block $j$, and
\beq
\ket{\omega_{1,\dots, p}} = \sum_{j=1}^p  \sigma_{x}^{[j]}  \ket{e_1} \otimes \cdots \otimes \ket{e_p} \ .
\eeq
Now we find upper and lower bounds to the maximum overlap of this state with a product state of the blocks. A lower bound can be easily obtained by the product state $\ket{0}^{\otimes (n-p)} \otimes  \ket{e_1} \otimes \cdots \otimes \ket{e_p}$. Noticing that $\ket{e_k}$ is either $\ket{0}$ or $\ket{+}$ \cite{tcmera}, and that the $\ket{+}$ contributions come only from qubits close to the boundary of the block, we have that

\beq
C|\widetilde{\mathcal{G}}|^{-1/2} \left(1 -  \frac{\omega n_b L \lambda}{\Delta} \right) \le \Lambda_{{\rm max}}^{{\rm [blocks]}} \  ,
\eeq
for some positive $\omega = O(1)$ constant. The following upper bound can also be found easily:
\beq
 \Lambda_{{\rm max}}^{{\rm [blocks]}} \le C |\widetilde{\mathcal{G}}|^{-1/2} \left( 1 + \frac{n_b L \lambda}{\Delta} + \frac{\lambda}{\Delta}\right) \ .
 \eeq
Using the above bounds, one can check that for the GE we obtain, in the limit $\lambda \ll \Delta$ and $L \gg 1$,
 \beq
\left(\frac{1}{4} + \frac{2 \omega \lambda}{\Delta}\right) n_b L- 1 \ge E_{G}^{\lambda} \ge \left(\frac{1}{4} - \frac{2 \lambda}{\Delta}\right) n_b L - 1 \ .
 \eeq
The above equation is compatible with a leading change in the GE in
the pre-factor of the boundary law. Also, the fact that both bounds
leave the topological component $E_{\gamma} = 1$ untouched seems to
indicate that this is actually robust under the perturbation.
Moreover, implementing finite-$L$ corrections to these bounds
provides $O(L^{-\nu'})$ corrections, which is consistent with our previous claims.

{ In order to further check the above arguments, we have computed numerically the GE for the toric code with $n$ spins on the square lattice, exactly, for sizes up to $n=16$, and where we added a perturbation that amounts to introducing a string tension in the Hamiltonian. More specifically, we perturbed the system by considering the PEPS tensors of the Toric Code ground state $\ket{0,0}$ \cite{TCPEPS} and modifying the non-zero components that correspond to having ``spin up'' in all sites. That is, the components $A_{\alpha \beta \gamma \delta}^i$ of the perturbed-PEPS tensors are given by

\begin{eqnarray}
A_{1111}^\uparrow = 1+g; &~~~~& A_{2211}^\downarrow = 1\nonumber \\
A_{2222}^\uparrow = 1+g; &~~~~& A_{1122}^\downarrow = 1 \nonumber \\
\end{eqnarray}
where the upper tensor index is the physical index and the rest are the bond indices (in the PEPS there is also another tensor like this but rotated 90 degrees, so that we have an $ABAB...$ periodicity with a 2-site unit cell). For $g=0$ one recovers the unperturbed toric code, whereas for $g=\infty$ one has a polarised state in the z-direction. This is the same kind of perturbation that was considered previously in Ref.\cite{chuang}, which adds a string tension with a somehow similar effect to adding a magnetic field in the z-direction to the Hamiltonian.

For such a perturbed topological state we did an exact calculation of the closest product state, for 1-site blocks (single spins) and also for 4-site blocks, up to $n=16$ (i.e. up to 16 1-site blocks, and up to 4 4-site blocks). From this we estimated $E_{\gamma}$ by considering the GE as a function of the number of blocks, fitting it to a straight line, and extrapolating the number of blocks to zero. Calculations for larger systems required a significant amount of computational resources, and therefore could not be implemented with this approach (for these one would need e.g. the tensor network techniques explained in the Appendix). A summary of our results can be found in Fig.~\ref{TopoGECorrect}. Here we see that in the case of single-site blocks, the topological GE is not stable under the perturbation: it drops very quickly to zero without any sign of phase transition instead of staying close to $E_{\gamma} = 1$  for a while. However, the 4-site block calculation shows that $E_{\gamma}$ also drops to zero, but quite slower than in the case of single-site blocks. This is a clear indication that in the of 4-site blocks the topological GE is more robust under perturbation. We expect, thus, that as the size of the blocks becomes larger, the topological contribution to the GE becomes more robust. At this point we are constrained by the sizes in the exact calculations, but we take these results as a first-principle indication that the topological contribution to the GE tends to be robust under perturbation for sufficiently large block sizes, which is in fact compatible with the RG and perturbation theory arguments above.}

\begin{figure}
 \includegraphics[width=9cm]{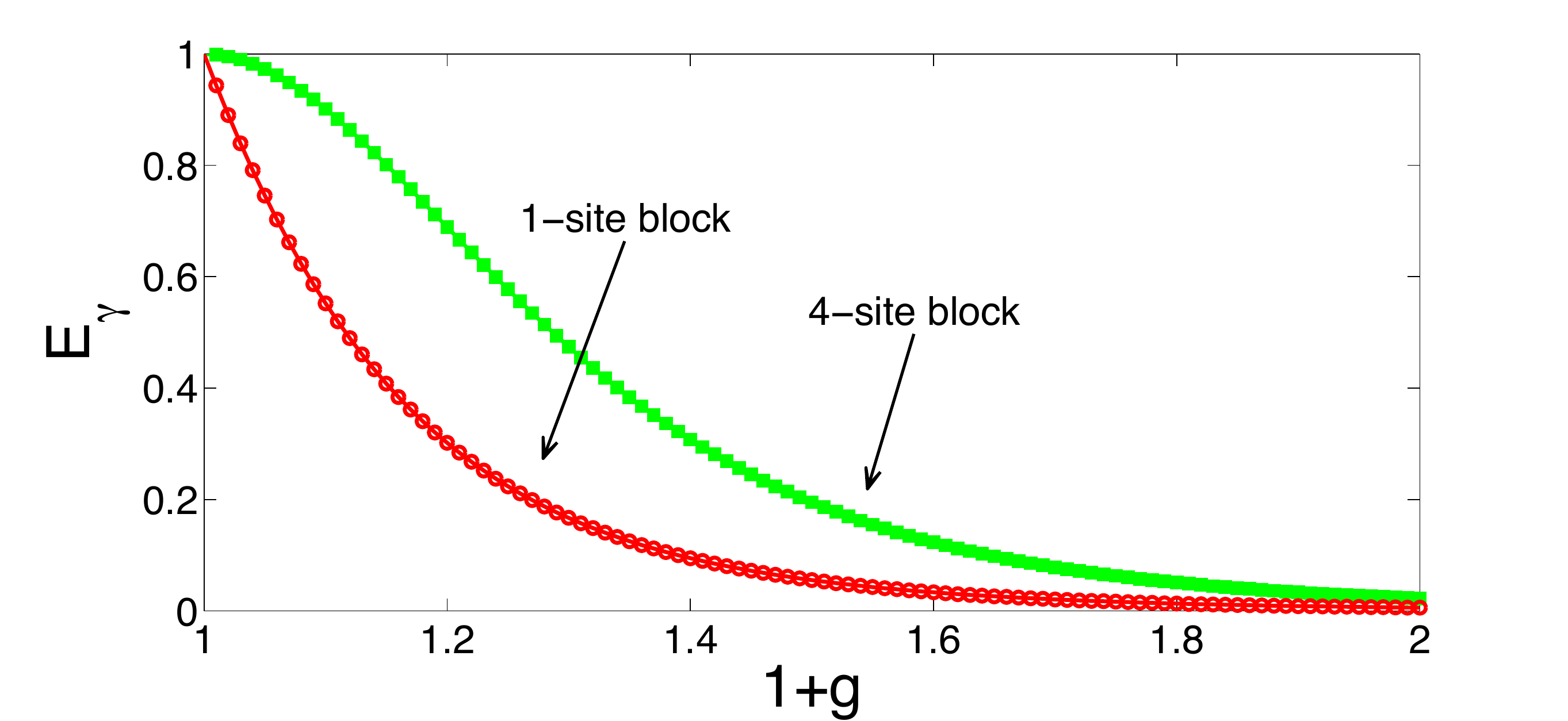}
  \caption{\label{TopoGECorrect} (color online) Estimated topological contribution to GE for the $\ket{0,0}$ ground state of the Toric Code on a square lattice, for systems up to $16$ sites, for 1-site (red) and 4-site (green) blocks.}
\end{figure}

\section{Conclusions}
\label{sec:conclude}

In this paper, we have studied the GE in
various topologically-ordered states that correspond to fixed-points of RG. These are the toric code, double semion, color code, and quantum double models, in a variety of lattices. Generically, we found that the GE is typically composed of a boundary law term times the number of blocks
plus a topological term. This is remarkable, since it is a signature of topological order in multipartite entanglement properties of the state, rather than bipartite. The topological term for the color code is
twice that for the toric code, which is consistent with the recent
result that the former model is equivalent to two copies of the
latter~\cite{Bombin2}. Away from RG fixed-points, we argued that the same type of behaviour holds
up to a possible subleading term, and thus the topological contribution is a robust property of topological phases of matter. The numerical evaluation of $E_\gamma$ for non-analytic cases is possible. This could be done e.g. in the context of Tensor Network methods (such as PEPS algorithms), as is usually done for the topological entropy \cite{numeric}. We leave { the specific implementation of these methods} for future work.

In all the cases considered here, the topological GE turns out to be equivalent to the topological entropy. It would be good to prove, in general, whether this is always true or not, in order to understand if the GE can provide more information about topological order than the one that is already in the topological entropy. However, even if this were not the case, the GE would still be a useful tool to extract the quantum dimensions of the anyon model associated to a given topologically ordered state. For scenarios where reduced density matrices are hard to evaluate numerically, one may thus believe that looking at the GE may be more efficient. { Moreover, the GE may also be a useful tool in order to extract \emph{minimally entangled states} within the ground subspace, and hence the complete topological characterisation of the system \footnote{Work in progress}.}

Finally, we believe that other models could possibly be analysed with the techniques that we used in this paper. Specifically,
similar results should also apply to the A-phase of Kitaev's honeycomb model \cite{honey} for
which the toric code is (in some limiting cases) an effective model,
as well as spin-liquid states with an emergent $\mathbb{Z}_2$ gauge
symmetry~\cite{z2}. The possibility of studying the ground states of
string-net models of Levin and Wen~\cite{sn} in general, as well as
topological quantum field theories \cite{tqft}, is also left for future investigation.

 \acknowledgements We thank M. Aguado, W.-M. Son, H.-H. Tu
 and G. Vidal for illuminating discussions and insightful
comments. Financial support from UQ, ARC, MPQ, EU, JGU, and NSF are
acknowledged. Part of this work was done at the Pedro Pasqual
Benasque Center for Science (CCBPP). This research was supported in
part by Perimeter Institute for Theoretical Physics. Research at Perimeter Institute
is supported by the Government of Canada through Industry Canada and
by the Province of Ontario through the Ministry of Research and Innovation. 
T.-C.W. acknowledges support by the National Science Foundation under Grant Nos. PHY 1314748 and PHY 1333903.

\appendix

\section{Relative entropy of entanglement}
\label{sec:REE} We shall see in the following that for most of the
ground states discussed in this paper, they can also be
characterised by the so-called relative entropy of entanglement and
in fact their value is identical to that of GE. The relative entropy
$S(\rho||\sigma)$ between two states $\rho$ and $\sigma$ is defined
via
\begin{equation}
S(\rho||\sigma)\equiv {\rm
Tr}\left(\rho\log_2\rho-\rho\log_2{\sigma}\right),
\end{equation}
which is evidently not symmetric under exchange of $\rho$ and
$\sigma$, and is non-negative, i.e., $S(\rho||\sigma)\ge 0$.  The
relative entropy of entanglement (RE) for a mixed state $\rho$ is
defined to be the minimal relative entropy of $\rho$ over the set of
separable mixed
states~\cite{VedralPlenioRippinKnight97,VedralPlenio98}:
\begin{equation}
\label{eqn:ER} E_R(\rho)\equiv \min_{\sigma\in {\cal
D}}S(\rho||\sigma)=\min_{\sigma\in {\cal D}}{\rm
Tr}\left(\rho\log_2\rho-\rho\log_2\sigma\right),
\end{equation}
where ${\cal D}$ denotes the set of all separable states. With any
separable state $\sigma$, we can obtain an immediate upper bound on
$E_R(\rho)\le {\rm Tr}(\rho\log_2\rho-\rho\log_2\sigma)$. For pure
state $\rho=|\Psi\rangle\langle\Psi|$, such as the ground state of
the toric code and the color code, we have further that
\begin{equation}
E_R(|\Psi\rangle\langle\Psi|)\le - \langle \Psi
|\log_2\sigma|\Psi\rangle.
\end{equation}

In general, the task of finding the RE for arbitrary states $\rho$
involves a minimisation over all separable states, and this renders
the computation of the RE very difficult. It has been shown that for
pure state $|\Psi\rangle$, the its GE lower bounds its REE, i.e.,
$E_R(\Psi)\ge E_G(\Psi)$~\cite{WeiErricsonGoldbartMunro}. Using
the technique in Ref.~\cite{WeiErricsonGoldbartMunro}, we show that the GE is identical to the REE for the ground states considered here.

Let us illustrate this by the $|0,0\rangle$ ground state of the
toric code on the square lattice. Since $|0,0\rangle=\sum_{g\in \mathcal{G}_s}
g|0...0\rangle/\sqrt{|\mathcal{G}_s|}$, where $\mathcal{G}_s$ is the group generated by star
operators, each component $|g\rangle\equiv g|0...0\rangle$ is a
product state. We can thus construct a mixed separable state
\begin{equation}
\sigma=\frac{1}{|\mathcal{G}_s|}\sum_{g\in \mathcal{G}_s}|g\rangle\langle g|,
\end{equation}
and hence obtain an upper bound on the REE,
\begin{eqnarray}
E_R(|0,0\rangle\langle0,0|) &\le& \log_2 |\mathcal{G}_s|\langle 0,0|\sum_{g\in \mathcal{G}_s}
\big(|g\rangle\langle g|\big)|0,0\rangle \nonumber \\
&=&\log_2|\mathcal{G}_s|.
\end{eqnarray}
This upper bound turns out to be identical to the value of the GE, a
lower bound on REE. Therefore, we have that
$E_R(|0,0\rangle\langle0,0|)=E_G(|0,0\rangle\langle0,0|)=n_s-1$.
Such an equality of GE and REE can also be understood from the group
symmetry~\cite{Hayashi}. We can straightforwardly apply similar
arguments to ground states of the double semion, color code and the quantum double
models.

{
\section{Some numerical tensor network algorithms to compute $E_{\gamma}$}
\label{sec:TNalg} In what follows we sketch, at a conceptual level, several methods based on tensor networks to numerically compute the topological GE. For concreteness we focus on the case of having a representation of the relevant topological quantum state given by a $2d$ Projected Entangled Pair State (or PEPS, see e.g. Ref.\cite{TCPEPS}). We assume that the system is finite, translationally-invariant, and is defined on the surface of a torus. We also assume that the size of the system is sufficiently large, so that relatively large blocks can be considered.

Within this framework, the topological GE can be extracted by e.g. doing a finite-size scaling of the GE of blocks, with respect to different block sizes. This GE of blocks needs to be computed by optimising the fidelity of the $2d$ PEPS on the torus with a product state of the blocks. The key point is how to implement this optimisation efficiently, considering that the number of sites within each block can be quite large, and that the total size of the system is also large.

Here two different strategies are proposed. These are based on renormalizing/not renormalizing the tensors within the blocks.

\vspace{10pt}

\underline{\emph{1.- Product-PEPS/MPS.-}} We consider a product state of the blocks such that the quantum state for each block is itself a finite $L \times L$ PEPS, or a (snake) Matrix Product State (MPS) for $L^2$ sites. The maximisation of the fidelity is then carried over the tensors of these finite-PEPS or finite-MPS, which can be implemented efficiently.

\vspace{10pt}

\underline{\emph{2.- Tensor-renormalization.-}} Another approach is to consider again blocks of size $L \times L$, but where we compute a renormalized PEPS tensor for the block. This could be achieved by using different tensor-renormalization strategies, such as e.g. SRG \cite{srg} and HOSRG \cite{hosrg} but adapted to a finite $2d$ system with open boundary conditions. As a result of this tensor-renormalization, the topological quantum state is represented by a ``renormalized'' PEPS where each tensor corresponds to a block. Once this is computed successfully (which may be non-trivial), the necessary optimisation to compute the GE  can be done using a product state directly for the renormalized single sites.


\begin{thebibliography}{100}

\bibitem{to}
X.-G. Wen, \emph{Quantum Field Theory of Many-Body Systems} (Oxford University Press, USA, 2004).

\bibitem{topoq}
See e.g. C. Nayak \emph{et al}, Rev. Mod. Phys. {\bf 80}, 1083 (2008).

\bibitem{tqft}
E. Witten, Comm. Math. Phys. {\bf 117}, 3,  353-386 (1988).

\bibitem{toric}
A. Y. Kitaev, Annals of Physics 303, 2-30 (2003).

\bibitem{sn}
M. A. Levin and X.-G. Wen, Phys. Rev. B {\bf 71}, 045110 (2005).

\bibitem{Hamma}
A. Hamma, R. Ionicioiu and P. Zanardi, Phys. Rev. A {\bf 71},
022315 (2005).

\bibitem{entr}
A. Kitaev and J. Preskill, Phys. Rev. Lett. {\bf 96}, 110404
(2006); M. Levin, X.-G. Wen, \emph{ibid} {\bf 96}, 110405 (2006).

\bibitem{mut}
S. Iblisdir \emph{et al.}, Phys. Rev. B {\bf 79}, 134303 (2009);

\bibitem{reny}
S. T. Flammia \emph{et al.}, Phys. Rev. Lett. {\bf 103}, 261601 (2009).

\bibitem{ge}
T.-C. Wei and P. M. Goldbart, Phys. Rev. A {\bf 68}, 042307 (2003).

\bibitem{brody}
D. C. Brody and L. P. Hughston, J. Geom. Phys. {\bf 38}, 19 (2001).

\bibitem{Bombin}
H. Bombin and M. A. Martin-Delgado, Phys. Rev. Lett. {\bf 97},
180501 (2006).

\bibitem{Bombin2}
H. Bombin, G. Duclos-Cianci, and D. Poulin,  New J. Phys. {\bf 14} (2012) 073048.

\bibitem{PEPS}
F. Verstraete and J. I. Cirac, Phys. Rev. A {\bf 70}, 060302(R) (2004); F. Verstraete, M. M. Wolf, D. Perez-Garc\'{\i}a, and J. I. Cirac, Phys. Rev. Lett. {\bf 96}, 220601 (2006).

\bibitem{MERA}
G. Vidal, Phys. Rev. Lett. {\bf 101}, 110501 (2008); M. Aguado and G. Vidal, Phys. Rev. Lett. {\bf 100}, 070404 (2008).

\bibitem{exper}
J. Zhang, T.-C. Wei and R. Laflamme, Phys. Rev. Lett. {\bf 107}, 010501 (2011).

\bibitem{opticalLattice}
A. J. Daley, H. Pichler, J. Schachenmayer and P. Zoller, Phys. Rev. Lett. {\bf 109}, 020505 (2012).

\bibitem{geometric2}
A. Botero and B. Reznik, arXiv:0708.3391; R. Or\'us, Phys. Rev. Lett. {\bf 100}, 130502 (2008); R. Or\'us, Phys. Rev. A {\bf 78}, 062332 (2008); T.-C. Wei, \emph{ibid} {\bf 81}, 062313 (2010).

\bibitem{geometric3}
T.-C. Wei \emph{et al}, Phys. Rev. A {\bf 71}, 060305(R) (2005); R. Or\'us, S. Dusuel and J. Vidal, Phys. Rev. Lett. {\bf 101}, 025701
(2008).

\bibitem{resource}
D. Gross, S. T. Flammia, and J. Eisert, Phys. Rev. Lett. {\bf 102}, 190501 (2009).

\bibitem{discrim}
M. Hayashi \emph{et al.}, Phys. Rev. Lett. {\bf 96}, 040501 (2006).

\bibitem{sc}
J. Eisert and M. Cramer, Phys. Rev. A {\bf 72}, 042112 (2005); R. Or\'us \emph{et al}, \emph{ibid} {\bf 73}, 060303(R) (2006); M. J. Bremner, C. Mora, and A. Winter, \emph{ibid} {\bf 102}, 190502 (2009).

\bibitem{cnots}
E. Dennis \emph{et al}., J. Math. Phys. 43, 4452--4505 (2002).

\bibitem{tcmera}
M. Aguado and G. Vidal, Phys. Rev. Lett. {\bf 100}, 070404 (2008).

\bibitem{boun}
See e.g. L. Amico \emph{et al.}, Rev. Mod. Phys.{\bf 80}:517-576 (2008).

\bibitem{Preskill}
See for example the discussion around Eq.(7.160) in Chapter 7 of the lecture notes of J. Preskill, www.theory.caltech.edu/people/preskill/ph229/

\bibitem{lwdis}
R. Koenig, B. W. Reichardt, G. Vidal, Phys. Rev. B {\bf 79}, 195123 (2009).

\bibitem{Va13}
M. Van den Nest and W. D\"ur, to appear (2013).

\bibitem{Kargarian}
M. Kargarian, Phys. Rev. A {\bf 78}, 062312 (2008).

\bibitem{numeric}
D. Poilblanc, N. Schuch, D. Perez-Garc\'{\i}a, J. I. Cirac, Phys. Rev. B {\bf 86}, 014404 (2012).

\bibitem{honey}
A. Kitaev, Annals of Physics. 321, 2 (2006).

\bibitem{z2}
R. Moessner, S. L. Sondhi and E. Fradkin, Phys. Rev. B {\bf 65}, 024504 (2001); S. V. Isakov, M. B. Hastings and R. G. Melko, Nature Physics doi:10.1038/nphys2036 (2011)

\bibitem{VedralPlenioRippinKnight97}
V. Vedral, M. B. Plenio, M. A. Rippin, and P. L. Knight, Phys. Rev.
Lett. {\bf 78}, 2275 (1997).

\bibitem{VedralPlenio98}
V. Vedral and M. B. Plenio, Phys. Rev. A{\bf 57}, 1619 (1998).

\bibitem{WeiErricsonGoldbartMunro}
T.-C. Wei, M. Ericsson, P. M. Goldbart, and W. J. Munro, Quantum
Inf. Comput. {\bf 4}, 252 (2004).

\bibitem{Hayashi}
M. Hayashi, D. Markham, M. Murao, M. Owari, and S. Virmani, Phys.
Rev. A {\bf 77}, 012104 (2008).

\bibitem{TCPEPS}
F. Verstraete, M. M. Wolf, D. P\'erez-Garc\'ia, J. I. Cirac, Phys. Rev. Lett. {\bf 96}, 220601 (2006).

\bibitem{chuang}
X. Chen, B. Zeng, Z.-Cheng Gu, I. L. Chuang, X.-Gang Wen, Phys. Rev. B {\bf 82}, 165119 (2010).

\bibitem{srg}
Z. Y. Xie, H. C. Jiang, Q. N. Chen, Z. Y. Weng, T. Xiang, Phys. Rev. Lett. {\bf 103}, 160601 (2009); H. H. Zhao, Z. Y. Xie, Q. N. Chen, Z. C. Wei, J. W. Cai, T. Xiang, Phys. Rev. B {\bf 81}, 174411 (2010).

\bibitem{hosrg}
Z. Y. Xie, J. Chen, M. P. Qin, J. W. Zhu, L. P. Yang, T. Xiang, Phys. Rev. B {\bf 86}, 045139 (2012).

\end{thebibliography}
\end{document}